\documentclass[
prx,
reprint,superstructures,superscriptaddress,longbibliography
]{revtex4-1}
\usepackage{bbm}
\usepackage{xcolor}
\usepackage{enumitem}
\usepackage{xfrac}
\usepackage{enumitem}
\usepackage{graphicx}

\usepackage{epstopdf}   
\usepackage{amssymb}
\usepackage{framed}
\usepackage{breakcites}

\usepackage{amsmath}
\usepackage{empheq}
\usepackage{fancybox}
\usepackage[latin1]{inputenc}
\usepackage{amssymb}
\usepackage{amsmath}
\usepackage{amsthm}
\usepackage{mathrsfs}
\usepackage{textcomp}

\usepackage{subcaption}

\usepackage[breaklinks=true]{hyperref}
\usepackage{breakcites}

\newtheorem{theorem}{Theorem}
\newtheorem*{theorem*}{Theorem}
\newtheorem{lemma}[theorem]{Lemma}

\newtheorem{definition}[theorem]{Definition}
\newtheorem{proposition}[theorem]{Proposition}

\let\originalleft\left
  \let\originalright\right
\renewcommand{\left}{\mathopen{}\mathclose\bgroup\originalleft}
  \renewcommand{\right}{\aftergroup\egroup\originalright}

\newcommand{\de}{d_{\mathrm{eff}}}          
\newcommand{\id}{\mathbbm{1}}                
\newcommand{\expo}[1]{\operatorname{e}^{#1}} 
\newcommand{\tr}[1]{\operatorname{Tr}\left[ {#1} \right]} 

\newcommand{\trs}[1]{\operatorname{Tr}_S\left[ {#1} \right]} 
\newcommand{\trb}[1]{\operatorname{Tr}_B\left[ {#1} \right]} 
\newcommand{\trx}[2]{\operatorname{Tr}_{#2}\left[ {#1} \right]} 
\newcommand{\ket}[1]{\left|#1 \right \rangle \vphantom{\left( #1 \right)^A}} 
\newcommand{\bra}[1]{\left\langle #1 \right | \vphantom{\left(#1\right)^A}} 
\newcommand{\proj}[1]{\ket{#1}\bra{#1}}

\newcommand{\SWAP}{\operatorname{\text{\textdollaroldstyle} }} 
\newcommand{\spn}{\operatorname{span}}
\newcommand{\rank}[1]{\relax\ifmmode\operatorname{rank}#1\else rank-$#1$\fi}

\allowdisplaybreaks[4]

\hypersetup{pdfborder={0 0 1}, bookmarks=true, linkcolor=blue}

\begin{document}

\title{Equilibration time scales of physically relevant observables}

\author{Luis Pedro Garc\'{i}a-Pintos}
\thanks{Corresponding author}
\email{lpgarciapintos@gmail.com}
\affiliation{School of Mathematics, University of Bristol, University Walk, Bristol BS8 1TW, U.K.}
\affiliation{Institute for Quantum Studies, Chapman University, 1 University Drive, Orange, CA 92866, USA.}

\author{Noah Linden}
\email{n.linden@bristol.ac.uk}
\affiliation{School of Mathematics, University of Bristol, University Walk, Bristol BS8 1TW, U.K.}
\author{Artur S.L. Malabarba}
\email{artur.malabarba@bristol.ac.uk}
\affiliation{H.H. Wills Physics Laboratory, University of Bristol, Tyndall Avenue, Bristol, BS8 1TL, U.K.}
\author{Anthony J. Short}
\email{tony.short@bristol.ac.uk}
\affiliation{H.H. Wills Physics Laboratory, University of Bristol, Tyndall Avenue, Bristol, BS8 1TL, U.K.}
\author{Andreas Winter}
\email{andreas.winter@uab.cat}
\affiliation{ICREA \& F\'{i}sica Te\`{o}rica: Informaci\'{o} i Fen\`{o}mens Qu\`{a}ntics, Universitat Aut\`{o}noma de Barcelona, ES-08193 Bellaterra (Barcelona), Spain }

\date{\today}
\begin{abstract}
We address the problem of understanding from first principles the conditions under which a quantum system equilibrates rapidly with respect to a concrete observable. On the one hand previously known general upper bounds on the time scales of equilibration were unrealistically long, with times scaling linearly with the dimension of the Hilbert space. 
These bounds proved to be tight, since particular constructions of observables scaling in this way were found.
On the other hand, the computed equilibration time scales for certain classes of typical  measurements, or under the evolution of typical Hamiltonians, turn out to be unrealistically short.
However neither classes of results cover physically relevant situations, which up to now had only been tractable in specific models. 
In this paper we provide a new upper bound on the equilibration time scales which, under some physically reasonable conditions, give much more realistic results than previously known.
In particular, we apply this result to the paradigmatic case of a system interacting with a thermal bath, where we obtain an  upper bound for the equilibration time scale  
independent of the size of the bath.
In this way, we find
general conditions that single out observables with realistic equilibration times within a physically relevant setup.
\end{abstract}

\maketitle

Knowing the details of how systems approach equilibrium is a major topic within statistical mechanics. 
However, deriving results on the equilibration time scales that are both general and apply to physically relevant situations has proven to be a challenge; one of the major open problems in understanding equilibration processes of quantum systems.

This paper addresses the time scales for reaching equilibrium in closed quantum systems. Recently there have been promising
advances~\cite{ShortFarrelly11,Goldstein13,
Malabarba14,
Goldstein14,Reimann16typical}, which add to the vast literature studying these issues in more specific models~\cite{Cramer08,CamposVenuti10,Kastner11, Vinayak12,Brandao12,Masanes13, Kastner13,CamposVenuti13,Fuchs13,
Zangara13,TorSan13,CamposVenuti15,
Banuls15,Kastner15} (for recent thorough reviews of this and related topics see~\cite{ReviewGogolin15}, \cite{EisertFriesdorfGogolin15} and~\cite{ReviewMarcusArnau15}).
In particular, we have learnt that typical  observables (when appropriately drawn at random) equilibrate rapidly~\cite{Goldstein13,Malabarba14,Reimann16typical}, and 
that the same is true for the evolution under typical Hamiltonians~\cite{Vinayak12,Brandao12,Masanes13} and for systems starting from typical non-equilibrium subspaces~\cite{Goldstein14}. Remarkably, this rapid equilibration has even been observed experimentally in certain systems~\cite{Reimann16typical}. 
Yet, one can construct observables which
take an extremely long time to approach equilibrium, up to a time proportional to the Hilbert space
dimension of the system~\cite{Goldstein13,Malabarba14}. 
Note that by fast vs slow equilibration we do not mean picoseconds vs years; slow can refer to timescales longer than the ``age of the universe''  for the constructions found in the papers mentioned above. %

It is important to note that the above mentioned results do not teach us a great deal about what happens for a given, physically relevant, observable. For instance, they do not answer what the time scales of equilibration for a system interacting with an environment are. Meanwhile, the typical (in the mathematical sense) measurements considered
will not generally represent physically relevant observables.

Moreover, the fact that one can always find mathematical constructions of observables which equilibrate in extremely long times, as in~\cite{Goldstein13} and~\cite{Malabarba14}, implies that extra -- more physical -- conditions are 
fundamental in singling out the observables which equilibrate within reasonable time scales.

In this paper, we consider the following physically relevant scenario -  measurements on a small system which is interacting with a large, highly mixed, bath via a given (non-random) Hamiltonian.  The main result will be to find sufficient conditions on initial state, observable and Hamiltonian that ensure reasonably fast time scales; in particular time scales that do not grow or decrease with the full dimension of the Hilbert space. 
We find that this is the case for sufficiently mixed initial states (as thermal states of not-too-low temperature are), provided some natural conditions on the off-diagonal matrix elements of the observable and initial state in the energy basis are met, that essentially ensure that a wide range of frequencies are involved in the evolution.   
This will be applied to the paradigmatic case of a small system interacting with a thermal bath in the microcanonical ensemble~\footnote{The calculations are greatly simplified in the microcanonical ensemble, although the results are also likely hold in the canonical ensemble, since the state of the bath is highly mixed for not-too-low temperatures in this case too.}, where we obtain an upper bound on the equilibration time scale which does not depend on the dimension of the bath.
Importantly, the results hereby obtained do not depend on particular details of the system under consideration.

We will say that a system equilibrates when it approaches some steady state, and
remains close to it, for some reasonably long time interval~\footnote{
 Notice that this is a more general process than approaching thermal equilibrium, since we do not restrict the equilibrium state to being thermal.
Moreover, we do not focus on the particular way in which the system approaches equilibrium and, in particular, the time decay to equilibrium need not be exponential. 
}.
Given the fact that for finite dimensional systems there always
exist revivals --- times (in general very long) in which the system returns arbitrarily close to
the initial state --- in quantum mechanics one cannot have equilibration in the strict sense. Therefore, following \cite{Lin09}, we will say that a system equilibrates if, for most times, its state is close to some fixed steady state. This fixed steady state is then called the equilibrium state. 

Here, this closeness is assessed with respect to some particular observable $A$, so we say equilibration takes place if $A$ cannot distinguish the instantaneous state from the equilibrium one. Restricting to different kinds of observables leads to different notions. Then, for instance, an observable acting on a subsystem probes whether that subsystem has equilibrated,
and what happens in the remainder of the closed system is only relevant in how it affects the evolution of this subsystem.
However, taking other kinds of observables, for example $A$ being some many body observable, gives a different view to the process. These sort of questions are particularly relevant since experiments are bringing mesoscopic quantum systems closer to observation~\cite{experiments1,experiments2,
experiments3,experiments4,experiments5}.  Notice that these situations are in general not described by master equations, and usually one needs to solve the actual evolution of the system in order to learn about time scales of equilibration.

We start in Section~\ref{sec:main result} by introducing the necessary notions for this paper, and a statement of the main result.
Section~\ref{sec:obsdepbound} contains
a general upper bound on the time averaged distance between instantaneous and equilibrium state, and an analysis 
of the time decay of this bound.
Section~\ref{sec:generaltheorem} contains an expression for the time scale of equilibration, which depends on the observable, state and Hamiltonian under consideration; the first main proof in the paper.
In~\ref{sec:microcanonical} we apply the result to the case of a system interacting with a thermal bath in the microcanonical ensemble, an important application of the previous part.
We end in Section~\ref{sec:unimodalityassumption} with an analysis of the conditions necessary to obtain reasonably fast equilibration. All detailed calculations can be found in the Appendices.

\section{Setting and special cases of the main results}
\label{sec:main result}

Consider a closed quantum system with a Hamiltonian $H$, and an initial state given by the density matrix $\rho_0$ in a Hilbert space $\mathcal{H}$. We start by focusing on a weak notion of distance between states, based on comparing the instantaneous expected value of an observable
$A$ to its equilibrium expected value,
\begin{equation}
\mathcal{\widetilde D}_A(\rho_t,\omega) = \frac{\big| \tr{\rho_t A} -\tr{\omega A} \big|^2}{4 \|A\|^2},
\end{equation}
where the evolved state is $\rho_t = e^{-iHt} \rho_0 e^{iHt}$, and $\omega = \langle \rho_t \rangle_{\infty}$ is the equilibrium state~\cite{Reimann08,Lin09}, where $\left\langle f(t) \right\rangle_T = \tfrac{1}{T} \int_0^T f(t) dt$ denotes a time average. Note that the equilibrium state is simply the initial state decohered in the energy basis, since the infinite time averaging removes any oscillating terms. The operator $A$ is assumed hermitian, with $\|A\|$ denoting its spectral norm~\footnote{The spectral norm is the maximum singular value, $||A|| = \max_{|x| \ne 0} \frac{|Ax|}{|x|}$, with $|x|$ the Euclidean norm for vectors.
We will also make use of the trace norm, defined by $||A||_1 = \tr{\sqrt{AA^\dag}}$, and the Hilbert-Schmidt norm $||A||_2 = \sqrt{\tr{AA^\dag}}$.}.
With this definition $0 \le \mathcal{\widetilde D}_A \le 1$. For simplicity we take units such that $\hbar = 1$.

Obviously, equilibration of expectation values does not imply equilibration in general, since one can have very different distributions with the same expected values. However, even for this weak notion of equilibration no reasonable time scale bounds for physically relevant observables were known up to now. Furthermore, it is easy to extend our calculations to a stricter notion of equilibration, the distinguishability between $\rho_t$ and $\omega$ given a measurement of $A$ (for completeness we show this in Appendix~\ref{app:distinguishability}). In order to distinguish the quantity $\mathcal{\widetilde D}_A(\rho_t,\omega)$ from the actual distinguishability,  we will call it the weak-distinguishability.

We can express the time average of the weak-distinguishability as
\begin{align}
\label{eq:dist1}
\left\langle  \mathcal{\widetilde D}_A(\rho_t,\omega) \right\rangle_T &= \frac{1}{4}\Bigg\langle \bigg| \sum_{j,k} e^{-i(E_j - E_k)t}(\rho_{jk} -\omega_{jk}) \frac{A_{kj}}{\|A\|} \bigg|^2 \Bigg\rangle_{T} \nonumber \\
 &= \frac{1}{4} \bigg\langle \Big| \sum_{\alpha} v_\alpha e^{-iG_\alpha t} \Big|^2 \bigg\rangle_{T} \nonumber \\
 &= \frac{1}{4} \sum_{\alpha \beta} v_\alpha v_\beta^* \left\langle e^{-i\left( G_\alpha -G_\beta \right) t} \right\rangle_T,
\end{align}
where energy levels are denoted by $E_j$, and the matrix elements of initial state, equilibrium state, and observable in the energy basis are $\rho_{jk}$, $\omega_{jk}$, and $A_{jk}$ respectively~\cite{ShortFarrelly11}. The index $\alpha$ represents pairs $(j,k)$ of levels  with distinct energies; we denote the corresponding energy gap by $G_\alpha = (E_j - E_k)$, and define the coefficients \begin{equation} 
v_\alpha = v_{(j,k)} = \rho_{jk} \frac{A_{kj}}{\|A\|}. 
\end{equation}
Notice that only terms with non-zero energy gaps appear in the sum in (\ref{eq:dist1}) since $\omega_{jk} = \rho_{jk}$ for $E_j = E_k$.

Our aim is to prove that the time average of the weak-distinguishability considered above becomes small. Since $\mathcal{\widetilde D}$ is a positive quantity this would allow us to conclude that for most times the weak-distinguishability is small, showing equilibration occurs. 
The main objective of this paper is to determine, or at least to upper bound, the time scale $T_{\text{eq}}$ in which this decay happens. 

The following normalised distribution will be crucial for our proofs:
\begin{align}
\label{eq:defdistribution}
p_{\alpha} \equiv \frac{|v_\alpha|}{Q} = \frac{1}{Q} \frac{ \big| \rho_{jk} A_{kj} \big|  }{\|A\|},
\end{align}
with the normalization factor
\begin{equation}
Q \equiv \sum_\alpha |v_\alpha| = \sum_{jk: E_j \ne E_k} |\rho_{jk}| \frac{|A_{kj}|}{\|A\|}.
\end{equation}
The distribution $p_{\alpha}$ contains information of all the physical quantities relevant for the dynamics, namely the observable $A$, initial state $\rho_0$, and the Hamiltonian $H$, and is a measure of the significance of the different energy gaps $G_{\alpha}$. 

Our main technical result is a general bound on equilibration times
for observables when the initial state is highly mixed (Theorem  \ref{theorem-genboundtime}).
Before embarking on the proofs of our general technical results, it
may be illuminating to see how they apply in certain special cases
that are of physical interest.  The first concerns a small system
interacting with a bath that is in a maximally mixed state.  The
second is a version of our main physical theorem (Theorem  \ref{theorem-boundmc}), in which
the bath is in a microcanonical state.

Let us first consider a  small system $\mathcal{S}$ of dimension $d_S$ interacting with a large bath in the maximally mixed state  $\rho_B = \frac{\id_B}{d_B^\Delta}$. We can then prove the following (this follows straightforwardly from Theorem~\ref{theorem-genboundtime} by taking $A = A_S \otimes \id_B$ and $Q$ bounded by (\ref{eq:simpleQbound})).
\begin{theorem}[Bound for system interacting with maximally-mixed bath]\label{thm:simplethem}
For any system observable $A = A_S \otimes \id_B$, initial state $\rho_0 = \rho_S \otimes \frac{\id_B}{d_B^\Delta}$, Hamiltonian $H = H_B + H_S + H_I$
\begin{align}
\label{eq:boundmcMainResult}
      \left\langle \mathcal{\widetilde D}_{A} \left( \rho_t ,\omega  \right) \right\rangle_T
&\le \frac{ \pi \, a(\epsilon) \|A_S\|^{1/2} Q^{5/2} }{T \sqrt{ \bigg| \textnormal{Tr}\Big( \Big[ \big[\rho_0,H_S\!+\!H_I \big] , H_S\!+\! H_I \Big] A_S \Big) \bigg|} } \nonumber \\
      &\quad + \pi \, \delta(\epsilon) \, Q^2, 
  \end{align}
where
\begin{align}
        \label{eq:boundmc-QfactorMainResult}
        Q &\le  \sqrt{ d_S \trs{\rho_S^2}}.
\end{align}
and  $a(\epsilon)$ and $\delta(\epsilon)$ depend on the distribution $p_\alpha$ and an arbitrary parameter $\epsilon>0$. They are described briefly below, and defined  in Proposition \ref{proposition-xidependence}. 
\end{theorem}

Crucially, we will show in Section~\ref{sec:obsdepbound} that if the initial state $\rho_0$, observable $A$, and Hamiltonian $H$ are such that $p_\alpha$ is spread over many different energy gaps and approximately unimodal then we can choose $\epsilon$ such that $\delta(\epsilon) \ll 1$ and $a(\epsilon) \sim 1$.
We will argue that this is to be expected for a wide range of initial states in systems with interacting Hamiltonians, 
and in Appendix~\ref{app:simulation} show it explicitly in a simulation of a spin ring, i.e. a 1D Ising model
with transversal magnetic field and periodic boundary conditions.

Moreover, we will argue in Section~\ref{sec:subsys-equil-mc} that we would expect to achieve a reduction in $\delta(\epsilon)$ as the size of the bath increases, hence the second term in eq.~\eqref{eq:boundmcMainResult} becomes small for large baths, and we obtain that equilibration  occurs for large enough times $T$.

We can think of  Theorem~\ref{thm:simplethem} as describing the system coupled to an infinite-temperature bath. To extend the analysis to a more physically realistic finite temperature bath (with inverse temperature $\beta$), we consider a bath  which is initially in the microcanonical ensemble. Hence, the bath starts in a state $\rho_B = \frac{\id_B^\Delta}{d_B^\Delta}$, where $\id_B^\Delta$ is the projector on some microcanonical window of width $\Delta$ and dimension $d_B^\Delta$.

We can then prove the following (this is Theorem~\ref{theorem-boundmc} in Section~\ref{sec:generaltheorem}, applied to the special case in which $A=A_S \otimes \id_B$. Theorem~\ref{theorem-boundmc} also applies to general observables). 
\begin{theorem}[Bound for system interacting with thermal bath]
For any system observable $A = A_S \otimes \id_B$, initial state $\rho_0 = \rho_S \otimes \frac{\id_B^\Delta}{d_B^\Delta}$, Hamiltonian $H = H_B + H_S + H_I$, and any $K > 0$ and $\epsilon>0$, the weak-distinguishability satisfies
\begin{align}
      \left\langle \mathcal{\widetilde D}_{A} \left( \rho_t ,\omega  \right) \right\rangle_T
&\le \frac{ \pi \, a(\epsilon) \|A_S\|^{1/2} Q_2^{5/2} }{T \sqrt{ \bigg| \textnormal{Tr}\Big( \Big[ \big[\rho_0,H_S\!+\!H_I \big] , H_S\!+\! H_I \Big] A_S \Big) \bigg|} } \nonumber \\
      &\quad + \pi \, \delta(\epsilon) \, Q_2^2  + \frac{18 }{K^2}, 
    \end{align}
where $a(\epsilon)$ and $\delta(\epsilon)$  are defined  in Proposition \ref{proposition-xidependence} . For a bath with density of states proportional to $ e^{\beta E}$ in the vicinity of the   microcanonical window, 
\begin{align}
        Q_{2} &\le   \sqrt{\frac{d_S \trs{\rho_S^2}\, e^{ \beta \|H_S\| +  \big( 1 + \sqrt{2d_s} \big) K \beta \|H_I\| }}{\left(1-\tfrac{1}{K}\right)\left(1-e^{-\beta \Delta} \right)}}+ \frac{2}{K}.
\end{align}
\end{theorem}

 Considering a sufficiently large bath and choosing the constant $K$ such that the last term is also small, we obtain that equilibration eventually occurs for large enough times $T$.

More precisely, we find that the system will be equilibrated with respect to $A$, in the sense described 
above,
for times $T \gg T_{\text{eq}}$, where 
\begin{equation} 
 T_{\text{eq}} \equiv \frac{\pi \, a(\epsilon) {\|A_S\|}^{1/2} Q_2^{5/2}}{\sqrt{\bigg| \text{Tr}\Big( \Big[ \big[\rho_0,H_S\!+\! H_I \big] , H_S\!+\! H_I\Big] A_S \Big) \bigg|}}. \label{eq3:genboundtimeestimate}
\end{equation}
Crucially, note that if interactions between the system and bath are short-range, such that $H_I$ only couples the system to a finite region in the bath (e.g. nearest neighbour interactions in a spin lattice) then  $ T_{\text{eq}}$ does not scale with the size of the bath. Instead it depends on details of the system and its coupling to the environment, and can be easily calculated from the initial state, observable and Hamiltonian once $a(\epsilon)$ has been estimated.

In the next section we will show a general bound on the weak-distinguishability, setting the ground for proving the above results.
 
\section{General bound on average distance}
\label{sec:obsdepbound}

Since $\mathcal{\widetilde D}_A$ is a positive quantity, it satisfies
\begin{align}
\label{eq:dist2}
\left\langle  \mathcal{\widetilde D}_A(\rho_t,\omega) \right\rangle_T \le \frac{5 \pi}{4} \left\langle  \mathcal{\widetilde D}_A(\rho_t,\omega) \right\rangle_{L_T},
\end{align}
where $\langle f(t) \rangle_{L_T} \equiv \int_{-\infty}^\infty \frac{f(t) T}{ T^2 + (t - T/2)^2} \frac{dt}{\pi}$ denotes the Lorentzian time average of the function 
$f$~\cite{Malabarba14}. 
Upper bounding the sum in (\ref{eq:dist1}) by taking the absolute value of all of the terms, incorporating (\ref{eq:dist2}), and using the fact  that $\left| \langle e^{i \nu t} \rangle_{L_T} \right| = e^{-|\nu| T}$,  we get
\begin{align}
\label{eqapp:xi}
\left\langle  \mathcal{\widetilde D}_A(\rho_t,\omega) \right\rangle_T \le \frac{5 \pi Q^2}{16} \sum_{ \alpha \beta   }   p_\alpha p_\beta e^{-|G_\alpha -G_\beta| T},
\end{align}
with  
$p_{\alpha} \equiv \frac{|v_\alpha|}{Q}$ and $Q \equiv \sum_\alpha |v_\alpha|$.

We are interested in the decay in time of $\left\langle \mathcal{\widetilde D}_A(\rho_t,\omega) \right\rangle_T$. For the normalized probability distribution $p_\alpha$ we define the function $\xi_p$ as follows. 
\begin{definition}[ ]
    \label{definition-xi}
Given any normalized probability distribution $p$ over the values of a real variable $Y$, we define $ \xi_p \left( x \right)$ as the maximum probability of any interval of length $x$. In particular, when $Y$ is discrete,
   \begin{align}
   \xi_p \left( x \right) = \max_{y_0 \in \mathbb{R}} \sum_{\alpha: y_\alpha \in \big[ y_0, y_0 + x \big]}  p_\alpha.
   \end{align}
\end{definition}

We prove in Appendix~\ref{app:xifunc} the following.
\begin{proposition}[General bound]
    \label{proposition-xibound}
    For any initial state $\rho_0$, any Hamiltonian, and any
    observable $A$,
    \begin{equation}
        \label{eq:genbound}
        \left\langle  \mathcal{\widetilde D}_A(\rho_t,\omega) \right\rangle_T \le \pi Q^2 \xi_p \big( \tfrac{1}{T} \big).
    \end{equation}
\end{proposition}

The function $\xi_p(x)$ will in general be difficult to compute explicitly, but for small $x$ it can be bounded (and well approximated) by a linear function. We will capture this behaviour in the following.
\begin{proposition}
    \label{proposition-xidependence}
    For any distribution $p$,
    \begin{align}
    \label{eq:xibound0}
        \xi_p \left( x \right) \le \frac{\xi_p(\epsilon)}{\epsilon} x + \xi_p(\epsilon), \qquad \forall \epsilon \in \left(0,\infty \right).
    \end{align}
It will be convenient to re-express this as
\begin{align}
\label{eq:xibound}
    \xi_p \left( x \right) \le \frac{a(\epsilon)}{\sigma} x + \delta(\epsilon),
\end{align}
where $\sigma$ is the standard deviation of the distribution, and we define 
\begin{equation}
\label{eq:aanddelta}
a(\epsilon) = \frac{\xi_p(\epsilon)}{\epsilon} \sigma, \qquad \delta(\epsilon) = \xi_p(\epsilon).
\end{equation}
\end{proposition}

\begin{proof}
Take $ \left( n -1 \right) \epsilon \le x < n \epsilon$, with $n \ge 1$ a natural number. The function $\xi_p$ is non-decreasing, hence 
$\xi_p \left( x \right) \le \xi_p \left( n \epsilon \right)$.
Since $\xi_p(\epsilon)$ quantifies the maximum probability that can fit \emph{any} interval $\epsilon$, we also have that $\xi_p(n \epsilon) \le n \xi_p(\epsilon)$, which results in
\begin{align}
\xi_p \left( x \right) &\le \xi_p \left( n \epsilon \right) \le (n-1) \xi_p( \epsilon) + \xi_p(\epsilon) \nonumber \\
&\le \frac{\xi_p(\epsilon)}{\epsilon} x + \xi_p(\epsilon).
\end{align}
\end{proof}

We now derive some general properties of $\xi_p$. 
For many distributions $p$ we would expect to be able to find an $\epsilon$ such that  $a(\epsilon) \sim 1 $ (in terms of its approximate order of magnitude)  and $\delta(\epsilon) \ll 1$.
To visualize how this can be so, consider the case in which the distribution has essentially a single ``peak'', and that the standard deviation $\sigma$ approximately quantifies the width of this peak. 
In such a case, a rough estimate for the maximum probability that can fit inside an interval $\epsilon$ can be given by
\begin{equation}
\label{eq:xiestimate}
\xi_p(\epsilon) \sim \frac{\epsilon}{\sigma}.
\end{equation}
With this estimate we indeed get $a(\epsilon) \sim 1$. Figure~\ref{fig:discretedistr} illustrates this for the case of a binomial distribution, where $0.2<a(\epsilon)<0.8$ for all $\epsilon>\frac{1}{2} $. 

\begin{figure}
     \centering
     \includegraphics[width=0.45\textwidth]{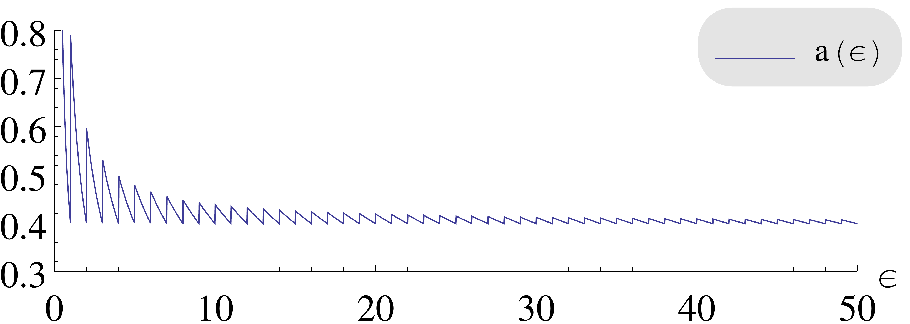}
     \includegraphics[width=0.45\textwidth]{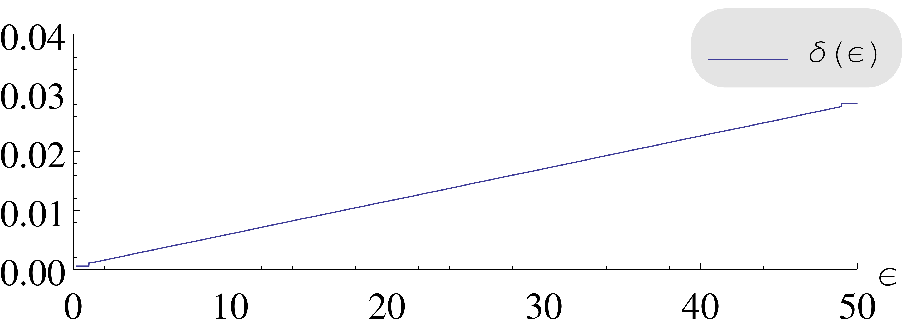}
     \caption{$a$ and $\delta$ as functions of $\epsilon$ for a binomial distribution of  $2 \times 10^6$ randomly chosen bits (mean $10^6$, standard deviation $\approx 707$).}
     \label{fig:discretedistr}
\end{figure} 

In general the above will work when the distribution $p$ is approximately unimodal, i.e. characterised by a single distinct peak. If, on the contrary, the distribution was composed of two or more peaks the estimate in equation~(\ref{eq:xiestimate}) might not hold, as Figure~\ref{fig:unimodalvsbimodal} exemplifies.

\begin{figure}
     \centering
     \includegraphics[width=0.5\textwidth]{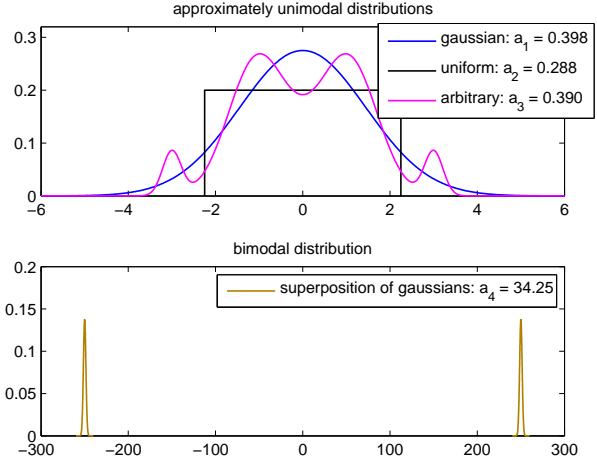}
     \caption{Examples of approximately unimodal distributions (above), and bimodal distribution (below), for continuous distributions in the limit $\epsilon \rightarrow 0$. The bimodal distribution can violate the estimate $\xi_p(\epsilon) \sim \frac{\epsilon}{\sigma}$, simply because one can make the standard deviation arbitrarily large by placing the peaks further apart without changing the actual value of $\xi_p(\epsilon)$.}
     \label{fig:unimodalvsbimodal}
\end{figure} 

When (\ref{eq:xiestimate}) holds, taking $\epsilon \ll \sigma$ is also enough to ensure $\delta(\epsilon) \ll 1$. Note that in Figure 1, $a(\epsilon)$ diverges for small $\epsilon$. To avoid such behaviour, we would typically want to choose $\epsilon$ larger than the gaps between consecutive values of the variable. Overall, we would expect to be able to find an $\epsilon$ satisfying both  $a(\epsilon) \sim 1 $ and $\delta(\epsilon) \ll 1$ if the distribution is approximately unimodal and 
spread over many different values of the random variable. 

In our particular case, Proposition~\ref{proposition-xibound} refers to the distribution $p_\alpha = p_{(j,k)} = \frac{|\rho_{jk}| |A_{kj}|}{Q \|A\|}$, which depends strongly on the distribution of energy gaps of the system. For large systems with typical energy ranges (e.g. finite positive temperatures), their energy levels tend to be more densely packed for larger energies, which leads to a much larger concentration of small gaps than large gaps. For most $A$ and $\rho_0$ we would therefore expect the distribution $p_\alpha$ to be more peaked towards the center and decay for larger values of the energy gaps $G$, leading to an approximately unimodal distribution over a dense spectrum as considered above. Nevertheless, this will not always be the case, as we will discuss in Section~\ref{sec:unimodalityassumption}.

\section{Observable dependent \protect\\ time scale bound}
\label{sec:generaltheorem}

Propositions~\ref{proposition-xibound} and~\ref{proposition-xidependence} lead to the following result.
\begin{theorem}[Observable dependent bound]
    \label{theorem-genboundtime}
Given an initial state $\rho_0$, observable $A$, Hamiltonian $H$, and any $\epsilon > 0$, the time averaged weak-distinguishability satisfies
    \begin{align}
        \label{eq3:genboundtime}
        \left\langle \mathcal{\widetilde D}_A(\rho_t,\omega) \right\rangle_T
& \le \pi \, Q^2 \left( \frac{a(\epsilon)}{\sigma_G T}  + \delta(\epsilon) \right) \nonumber\\
       & \le \frac{ \pi \, a(\epsilon) {\|A\|}^{1/2}  Q^{5/2}}{T \sqrt{\Big| \textnormal{Tr}\Big( \big[ \left[\rho_0,H \right] , H\big] A \Big) \Big|}} + c\, \delta(\epsilon) Q^2,
    \end{align}
    where
    \begin{align}
        \label{eq3:genboundtime-Qfactor}
        Q^2 \le d \tr{ \rho_0^2 },
    \end{align}
 $\sigma_G$ is the standard deviation of energy gaps $G_{\alpha}$ for the distribution $p_{\alpha}$, $a(\epsilon)$ and $\delta(\epsilon)$ are as in Proposition~\ref{proposition-xidependence}, 
 and $d$ is the rank of $\omega$. 
\end{theorem}
\begin{proof}
Since the distribution $p_\alpha$ is symmetric with respect to interchanging the indices $\left\{j,k\right\}$ while $G_\alpha$ is antisymmetric, we get that its variance, denoted by $\sigma_G^2$, satisfies 
\begin{align}
  \label{eq:app-sigma}
  \sigma_G^2 &=\sum_\alpha p_\alpha G_\alpha^2 - \left( \sum_\alpha p_\alpha G_\alpha \right)^2
               =\sum_\alpha p_\alpha G_\alpha^2 \notag\\
             &= \sum_{jk} \frac{|\rho_{jk}|  |A_{kj}|}{Q\|A\|} (E_j-E_k)^2 \notag\\
             &\ge \frac{1}{Q\|A\|}\Big| \sum_{jk} \rho_{jk} A_{kj} (E_j-E_k)^2 \Big| \notag\\
             &= \frac{1}{Q\|A\|}\Big| \text{Tr}\Big(\big[ \left[\rho_0 ,H \right] , H\big] A \Big) \Big|. 
\end{align}
Notice that for a local Hamiltonian and observable, and a known initial state, this expression (combined with the bound for $Q$ which soon follows) is much simpler to compute than $\sigma_G$, since it does not require detailed knowledge of the Hamiltonian's spectrum and eigenbasis, which is needed in order to construct the distribution $p_\alpha$ in the first place.

Moreover, we find an upper bound for $Q$,
\begin{align}
\label{eq:app-Qbound}
Q^2 &= \Bigg( \sum_{jk: E_j \ne E_k} |\rho_{jk}| \frac{|A_{kj}|}{\|A\|} \Bigg)^2 \nonumber \\
&\le \Bigg( \sum_{jk} |\rho_{jk}| \frac{|A_{kj}|}{\|A\|} \Bigg)^2 \nonumber \\
 &\le \Bigg( \sum_{jk} |\rho_{jk} |^2 \Bigg) \Bigg( \sum_{j'k' \in \ \text{supp}(\omega)} \frac{|A_{j'k'}|^2}{\|A\|^2} \Bigg) \nonumber \\
 &= \tr{\rho_0^2} \frac{\tr{\Pi_{\omega} A \Pi_{\omega} A}}{\|A\|^2} \nonumber \\
 &\le \tr{\rho_0^2} \tr{\Pi_{\omega}}  = d \tr{\rho_0^2} ,
 \end{align}
where $\Pi_{\omega}$ projects onto all the energy levels of the Hamiltonian which occur with non-zero probability in $\rho_0$  (this is given by the support of $\omega$). In the third line we restrict to this set of energy levels and use the Cauchy-Schwarz inequality. For the last line, notice that  by the Cauchy-Schwarz inequality for the scalar product $\left( O , P \right) \equiv \tr{O P^\dag}$ we can bound $\tr{\Pi_{\omega} A \Pi_{\omega} A} \le  \tr{\Pi_{\omega} A^2}$. Then using that for any two positive semidefinite matrices $\tr{O P} \le \|O\| \tr{P}$ we find $\tr{\Pi_{\omega} A^2}~\le~\tr{\Pi_{\omega}}~\|A^2\|~=~d \|A\|^2$.

Inserting the above into  eqs.~(\ref{eq:genbound}) and~(\ref{eq:xibound}) 
proves our claim.
\end{proof}
    
If the second term on the right hand side of equation~(\ref{eq3:genboundtime}) is small, the system will eventually equilibrate with respect to $A$. 
The time dependence is determined by the first term.
In particular, the system will be equilibrated (in the sense described in 
Section~\ref{sec:main result}) for times $T \gg T_{\text{eq}}$, where 
\begin{equation} 
 T_{\text{eq}} \equiv \frac{\pi \, a(\epsilon) {\|A\|}^{1/2} Q^{5/2}}{\sqrt{\Big| \text{Tr}\Big( \big[ \left[\rho_0,H \right] , H\big] A \Big) \Big|}}. \label{eq3:genboundtimeestimate}
\end{equation}

It is interesting to note the dependence of the above expression on $\big[ [\rho_0,H], H \big]$, which is, up to a minus sign, the second time derivative of the state at $t=0$. Therefore $T_{\text{eq}}$ can be alternatively written as
\begin{equation}
T_{\text{eq}} = \frac{ \pi \, a(\epsilon) {\|A\|}^{1/2} Q^{5/2}} {\sqrt{\Big|   \tr{ \left. \frac{d^2 \rho_t}{dt^2} \right|_{t=0} A } \Big|}}.
\end{equation}
Remarkably, the denominator of this expression is what one would expect from a Taylor expansion of the distance for short times, assuming the system is initially as far from equilibrium as possible (then the first derivative term is $0$ and one is left with the second derivative as leading order). 

We argued earlier that we would typically expect $a(\epsilon) \sim 1$ and $\delta(\epsilon) \ll 1$. However, we still have to address the size of the bound for $Q$ given by equation~(\ref{eq3:genboundtime-Qfactor}), which could greatly influence the speed of equilibration.
Notice that in general the dimension $d$ of the Hilbert space is extremely large, since it scales exponentially with the number of constituents of the closed system being considered.
Therefore, in order for this bound to show rapid equilibration we would need a very mixed initial state,  spread over a significant fraction of the  Hilbert space. 

Moreover, the constant $Q$ appears in the second term in Theorem~\ref{theorem-genboundtime}, along with $\delta(\epsilon)$. In order to show equilibration at all this second term needs to be small too.

In the next section we consider an important physical scenario and then use our bound to show reasonably fast equilibration.

\section{System interacting with a
bath}
\label{sec:microcanonical}

We now turn to the paradigmatic case of a small system interacting with a large thermal bath. This situation corresponds to decomposing the closed system considered in the previous sections into a small system $\mathcal{S}$ and a bath $\mathcal{B}$.
By assuming the observable $A$ to be of the form
\begin{equation}
A = A_S \otimes \id_B,
\end{equation}
where $A_S$ acts on the system
and $\id_B$ is the identity acting on the bath, one can focus on the system's behaviour.

The total Hamiltonian is denoted by
\begin{align}
H = H_S + H_B + H_I,
\end{align}
where $H_S$ and $H_B$ are the system and bath Hamiltonians, and $H_I$ denotes the interaction between them.

We assume that the system $\mathcal{S}$ is initially in an arbitrary state $\rho_S$, and for simplicity not correlated with the initial state of bath $\rho_B$, that is
$\rho_0 = \rho_S\otimes\rho_B$,
corresponding to a system initially isolated which is suddenly allowed to interact with $\mathcal{B}$ via $H_I$.

To show that such a situation can lead to a small value for $Q$, we first consider the case in which the bath is in a maximally mixed state, with 
\begin{equation} 
\rho_0 = \rho_S \otimes \frac{\id_B}{d_B}.  
\end{equation}
In this case, it is easy to see from equation (\ref{eq3:genboundtime-Qfactor}) that 
\begin{equation} \label{eq:simpleQbound} 
Q\leq d \tr{\rho_0^2} \leq (d_S d_B)\trs{\rho_S^2} \trb{\frac{\id^2_B}{d_B^2}} = d_S \trs{\rho_S^2}.
\end{equation} 

The remainder of this section corresponds to extending this simple example (which could be understood as a system interacting with an infinite temperature bath) to the more physical case of a system interacting with a finite temperature bath. 

In what follows, given a Hamiltonian $H$ we denote an energy window of width $\Delta$ centered around an energy
$E$ in terms of its corresponding Hilbert space
$\mathcal{H}^{E,\Delta}_{H}$, defined as
\begin{equation}
    \label{eq:1x}
\mathcal{H}^{E,\Delta}_{H} = \spn \left\{ \ket{E_n} : E-\frac{\Delta}{2} \le E_n \le E+\frac{\Delta}{2} \right\} .
\end{equation}

We will consider the state of the bath from the microcanonical ensemble viewpoint.
Consequently, we consider an  energy window of the bath Hamiltonian of width $\Delta$ centered around $E_B$. The subspace that this defines, $\mathcal{H}^{E_B,\Delta}_{H_B}$, will be referred to
as a \emph{microcanonical window}, and its dimension will be denoted by $d_B^\Delta$. The initial state of the bath is then $\rho_B = \frac{\id_B^\Delta}{d_B^\Delta}$, and the initial state of the system plus bath is
\begin{equation}
\rho_0 = \rho_S\otimes\frac{\id_B^\Delta}{d_B^\Delta},
\end{equation}
which we will use from now on.

The width of the microcanonical window is to be taken large enough such that it contains many
energy levels, in particular many more than the dimension of the
system, yet small in comparison to the whole spectrum of the bath
Hamiltonian $H_B$. 

\subsection{Truncation of the Hilbert space}
\label{sec:truncation}

Notice that the state $\rho_0 = \rho_S\otimes\frac{\id_B^\Delta}{d_B^\Delta}$, corresponding to a bath in the microcanonical ensemble, is quite mixed. This is good news for our bound $Q^2 \le d \tr{\rho_0^2}$ given by equation~(\ref{eq3:genboundtime-Qfactor}),
 since the purity of the state
will be a small number. However,
the presence of the dimension of the Hilbert space implies that the bound for $Q$ could still be extremely large.
In this section we show a truncation method for the state and the Hilbert space
 which allows us to reduce the relevant dimension significantly.

As the eigenvalues of $H_S$ lie between $-\|H_S\|$ and $\|H_S\|$, the initial state $\rho_0$ is contained
inside an energy window of width $\Delta + 2 \|H_S\|$ of the unperturbed Hamiltonian $H_S + H_B$.  
That is, $\rho_0$ lies within the subspace $\mathcal{H}^{E_B,\Delta+2\|H_S\|}_{H_B + H_S}$.

However, this no longer holds when considering the full Hamiltonian $H = H_S + H_B + H_I$; in principle $\rho_0$
can have support outside the subspace $\mathcal{H}^{E_B,\Delta+2\|H_S\|}_{H_B + H_S + H_I}$. 
Yet, one has the intuition that if the interaction
is small compared to the unperturbed Hamiltonian the energy window where the state is supported should not grow significantly.
This would imply that one does not need to
consider the full Hilbert space, but rather a truncated subspace
corresponding to a window somewhat larger than the original $\mathcal{H}^{E_B,\Delta+2\|H_S\|}_{H_B + H_S}$.

The above reasoning is proved correct in Appendix~\ref{app:truncation-errorsmc}, where we show that the trace distance between the state $\rho_0$ and a truncated state $\Pi \rho_0 \Pi $ is small, where $\Pi$ is a projector onto the truncated subspace. More precisely, 
we find the following.
\begin{proposition}[Hilbert space truncation]
    \label{proposition-truncationmc}
    For any $K$, the state $\rho_0 = \rho_S\otimes\frac{\id_B^\Delta}{d_B^\Delta}$ which lies within $\mathcal{H}^{E_B,\Delta+2\|H_S\|}_{H_B + H_S}$ can be truncated to the state $\Pi \rho_0 \Pi$, with 
    \begin{align}
    \left\| \rho_0 - \Pi \rho_0 \Pi \right\|_1 & \le  \frac{2}{K}, \label{eq:truncateddistance}
    \end{align}
        where $\Pi$ projects onto the subspace $\mathcal{H}^{E_B,\Delta + 2\| H_S \| + \eta}_{H_B + H_S + H_I}$ with a width extended by $\eta = \sqrt{8 d_S} \|H_I\| K$.
\end{proposition}

As a straightforward corollary one obtains that 
\begin{align}
\mathcal{\widetilde D}_A\left( \rho_t, \Pi \rho_t  \Pi \right) & \le  \frac{1}{K^2}, \label{eq:truncateddistinguishability}
\end{align}
where $\rho_t  = e^{-iHt} \rho_0 e^{iHt}$ is the evolved state. This shows that,  as long as we take $K$ large enough, the two states give similar evolutions.

This truncation procedure will be particularly useful to us, since the dimension of the accessible Hilbert space $\mathcal{H}^{E_B,\Delta +2 \|H_S\| + \eta}_{H_B + H_S + H_I}$ is in general much smaller than the full 
 dimension. 
 
We also find in Appendix~\ref{app:truncation-errorsmc} that, if the density of states of the bath is denoted by $\nu_B(E)$, the dimension of the truncated state, $d_\text{trunc} = \rank{\left(\Pi\right)}$, satisfies
\begin{align}
d_\text{trunc} \le \frac{d_S}{1-\tfrac{1}{K}} \int_{E_B - \tfrac{\Delta}{2} -  \left( 1 + \sqrt{2d_s} \right) K \|H_I\| -\|H_S\|  }^{E_B + \tfrac{\Delta}{2} +\left( 1 + \sqrt{2d_s} \right) K \|H_I\| + \|H_S\|  } \nu_B(E) dE. 
\end{align}
Meanwhile, the dimension of the (unperturbed) microcanonical window of the bath is given by
\begin{align}
d_B^\Delta = \int_{E_B - \tfrac{\Delta}{2}}^{E_B + \tfrac{\Delta}{2}} \nu_B(E) dE.
\end{align}

Typically, thermal baths have a (coarse grained) density of states which grows approximately exponentially with energy. Thus, if we take
\begin{equation}
\nu_B(E) = \mathcal{N} e^{\beta E},
\end{equation}
where $\beta$ is the inverse temperature and $\mathcal{N}$ a normalization constant, it is easy to obtain
\begin{align}
\label{eq:truncateddimension}
&d_\text{trunc} \le \\
&  \frac{d_S d_B^\Delta}{\left(1-\tfrac{1}{K} \right)} \frac{ \sinh{ \Big[ \beta \frac{\Delta}{2} +  \beta \|H_S\| +  \big( 1 + \sqrt{2d_s} \big) K \beta \|H_I\| \Big] }}{ \sinh{ \big[ \beta \frac{\Delta}{2} \big] } } \nonumber \\
& \quad \le \frac{d_S d_B^\Delta}{\left(1-\tfrac{1}{K}\right)\left(1-e^{-\beta \Delta} \right)} e^{ \beta \|H_S\| +  \big( 1 + \sqrt{2d_s} \big) K \beta \|H_I\| }. \nonumber
\end{align}
Note that, given that the energy width of the microcanonical window grows as the number of constituents of the bath increases, in general $\beta \Delta \gg 1$ holds for a large enough bath, in which case the last inequality is a particularly good approximation.

\subsection{Time scales for a system in contact with
a bath}
\label{sec:subsys-equil-mc}

Proposition~\ref{proposition-truncationmc}  allows us to truncate the microcanonical state $\rho_0$ to $\Pi \rho_0 \Pi$, since the error introduced is small. This greatly reduces the dimension of the relevant Hilbert space, and consequently the corresponding bound for the constant $Q$ in Theorem~\ref{theorem-genboundtime}.

However, this reasoning would also lead us to use the truncated state in the theorem itself, which would cause  the  replacement of $\text{Tr}\Big( \big[ \left[\rho_0,H \right] , H\big] A \Big)$ by $\text{Tr}\Big( \big[ \left[\Pi \rho_0\Pi,H \right] , H\big] A \Big)$ in equation~(\ref{eq3:genboundtime}). This not only introduces additional complexity but could possibly significantly weaken the bound.
Moreover, even if the Hamiltonian involved nearest neighbour-type
interactions, $\Pi$ could be highly non-local and indeed we may have no
way of computing it.
Nevertheless, we prove in Appendix~\ref{app:truncation-commutators} that the time average of the weak-distinguishability can be bounded with a commutator involving the original state $\rho_0$ instead of the truncated one, while still having a relevant Hilbert space with much smaller dimension than the original space. 

We finally have all the ingredients to apply Theorem~\ref{theorem-genboundtime} to the case of a system in contact with a thermal bath, which turns into the following.
\begin{theorem}[Bound for system interacting with thermal bath]
\label{theorem-boundmc}
For any $\epsilon>0,K > 0$, observable $A$, Hamiltonian $H = H_B + H_S + H_I$, and initial state $\rho_0 = \rho_S \otimes \frac{\id_B^\Delta}{d_B^\Delta}$, 
the weak-distinguishability satisfies
\begin{align}
\label{eq:boundmc}
      \left\langle \mathcal{\widetilde D}_{A} \left( \rho_t ,\omega  \right) \right\rangle_T
&\le \frac{\pi \, a(\epsilon) \|A\|^{1/2}  Q_2^{5/2} }{T \sqrt{\Big| \textnormal{Tr}\Big( \big[ \left[\rho_0,H \right] , H \big] A\Big) \Big|}} \nonumber \\
      &\quad + \pi\, \delta(\epsilon) Q_2^2  + \frac{18 }{K^2}, 
    \end{align}
    where $a(\epsilon)$ and $\delta(\epsilon)$ are as in Proposition~\ref{proposition-xidependence}, and
\begin{align}
	Q_2 &\le 				\sqrt{	\tr{ \rho_0^2} d_\textnormal{trunc}} + \frac{2}{K}, \\
d_\textnormal{trunc} &\le \frac{d_S}{1-\tfrac{1}{K}} \int_{E_B - \tfrac{\Delta}{2} -  \left( 1 + \sqrt{2d_s} \right) K \|H_I\| -\|H_S\|  }^{E_B + \tfrac{\Delta}{2} +\left( 1 + \sqrt{2d_s} \right) K \|H_I\| + \|H_S\|  } \nu_B(E) dE. \nonumber
\end{align}
where $ \nu_B(E)$ is the density of states of the bath.  Moreover, if we take the density of states of the bath to be $ \nu_B(E) \propto e^{\beta E}$ (as we would expect for a thermal bath in the vicinity of the microcanonical window), we obtain from eq.~(\ref{eq:truncateddimension}) that
\begin{align}
        \label{eq:boundmc-Qfactor}
        Q_{2} &\le   \sqrt{\frac{d_S \trs{\rho_S^2}\, e^{ \beta \|H_S\| +  \big( 1 + \sqrt{2d_s} \big) K \beta \|H_I\| }}{\left(1-\tfrac{1}{K}\right)\left(1-e^{-\beta \Delta} \right)}}+ \frac{2}{K}.
\end{align}
\end{theorem}
Taking a system observable $A = A_S \otimes \id_B$ we recover the main result in Section~\ref{sec:main result}, since the Hamiltonian $H_B$ in eq.~\eqref{eq:boundmc} commutes with $\rho_0$ and $A$.

Let us consider this result more closely. To begin with, all time independent terms have to be small for our theorem to imply equilibration in the first place.  The factors involving $K$ in equations~(\ref{eq:boundmc}) and~(\ref{eq:boundmc-Qfactor}) come from the truncation procedure, 
and are small as long the microcanonical window and the truncation window are large enough. 

The other time independent term is $\pi\, \delta(\epsilon) Q_{2}^2$, which we have neglected so far. As discussed in 
Section~\ref{sec:obsdepbound},
for distributions $p_\alpha$ that are approximately unimodal and sufficiently spread over different values one can  estimate that $\delta(\epsilon) \sim \frac{\epsilon}{\sigma_G} $. 
Notice that as the bath grows in size one would expect that this holds for smaller values of $\epsilon$, since the distribution $p_\alpha$ would be spread over more values. We could therefore take $\epsilon$ smaller and smaller and  reduce $\delta(\epsilon)$.
At the same time, the bound on $Q_{2}$ in equation~(\ref{eq:boundmc-Qfactor}) will generally not grow with the dimension of the bath. 
To see this note that typically (e.g. for short range interactions in a lattice system), $\|H_I\|$ will not increase significantly as the bath size increases, and that increasing the width $\Delta$ of the microcanonical window as the bath grows will cause the bound to become tighter.
Therefore, in the limit of increasing bath sizes the term $\pi\, \delta(\epsilon) Q_{2}^2$ becomes negligible, as needed. 

The fact that the results in Theorem~\ref{theorem-boundmc} do not depend on the dimension of the full Hilbert space is a very noticeable aspect of this paper. 
This is in stark contrast with previously know general upper bounds on the time scale of equilibration~\cite{ShortFarrelly11}, which essentially scale with the full Hilbert space dimension.

Finally, the first term in eq.~(\ref{eq:boundmc}) determines the time decay of the weak-distinguishability, and can be interpreted the same way as in the corresponding term in Theorem~\ref{theorem-genboundtime} (see subsequent discussion). 
Notice that, once $a(\epsilon)$ is estimated, the time dependence can in general be calculated analytically for a given initial state, Hamiltonian and observable. 
Moreover, performing this calculation is much simpler than solving the exact time evolution, which involves commutators of initial state and Hamiltonian of all orders and can only be done for simple models.

It is illuminating to ask how our bound behaves in a case where no equilibration occurs. Take, for example, a spin $1/2$ in a pure initial state $\ket{\Psi} = \tfrac{1}{\sqrt{2}} \left( \ket{\uparrow} + \ket{\downarrow} \right)$ as the system $\mathcal{S}$, and a bath composed of $N$ other spins in the microcanonical ensemble. Furthermore, take the Hamiltonian $H = \Omega \sigma_z^S + H_B$, and the observable $A = \sigma_x^S \otimes \id_B$. Since the system does not interact with the bath it does not equilibrate with respect to the observable $A$. The key to understanding where our bound expresses this fact is in the factor $\delta(\epsilon)$. It is easy to see that the distribution $p_\alpha$ is composed of only 2 values, corresponding to the gaps $\Omega$ and $-\Omega$, which
results in $\delta(\epsilon) \ge \tfrac{1}{2}$ for any $\epsilon$, hence no equilibration at all.

\subsection{System interacting with an environment in a pure state: the typical behaviour}
\label{sec:statebath}

So far we have considered mixed initial states of the total closed system. Here we show that our results can be extended to the typical behaviour of pure initial states of the environment that interact with the system.

Let us consider the environment's initial state to be pure, and drawn at random from the microcanonical window. Any pure state from the microcanonical window can be written as
\begin{align}
\rho_B^U = U \ket{\psi} \bra{\psi} U^\dag \label{eq:initialbath},
\end{align}
where $U$ is a unitary operator acting on $\mathcal{H}^{E_B,\Delta}_{H_B}$. By averaging over all possible $U$'s, drawn from the Haar measure, we have the typical behavior for random pure states from the microcanonical subspace $\mathcal{H}^{E_B,\Delta}_{H_B}$.

It turns out that taking the initial environment state to be a  pure state chosen at random from a microcanonical window  leads to very similar results to the environment starting in the microcanonical mixed state $\frac{\id_{B}^\Delta}{d_B^\Delta}$ in the subspace $\mathcal{H}^{E_B,\Delta}_{H_B}$.
More precisely, we show in Appendix~\ref{app3:typical-bath} the following.
\begin{proposition}[Evolution for typical initial states of the bath]
    \label{proposition-typicalbath}
    The weak-distinguishability averaged over all possible initial pure states of the environment drawn from a microcanonical window 
    of width $\Delta$ satisfies
    \begin{align}
    \left\langle \mathcal{\widetilde D}_A\big(\rho_t^U,\omega^U\big) \right\rangle_U \le \mathcal{\widetilde D}_A\big(\rho_t ,\omega \big)  +   \frac{d_S}{d_B^\Delta},       \nonumber
    \end{align}
where $\rho_0^U = \rho_S \otimes U \ket{\psi} \bra{\psi} U^\dag$ with corresponding evolved and equilibrium states $\rho_t^U$ and $\omega^U$, and $\rho_0 = \rho_S \otimes \frac{\id_{B}^\Delta}{d_B^\Delta}$ with corresponding evolved and equilibrium states $\rho_t $ and $\omega $.
\end{proposition}
Since the microcanonical window is assumed to contain many more levels than the system's dimension, $d_B^\Delta \gg d_S$, the above expression implies that for typical initial pure states of the bath the evolution is as if the initial state was $\rho_0 = \rho_S \otimes \frac{\id_{B}^\Delta}{d_B^\Delta}$. 

It is straightforward to combine this Proposition with Theorem~\ref{theorem-boundmc} and show that the upper bound for the typical time scale of equilibration for a system interacting with a bath in a pure state is the same as if the bath were in the microcanonical state.

\section{Discussion}

\label{sec:unimodalityassumption}

From previous work we know that one needs to impose further conditions in order to prove reasonably fast equilibration, since extremely slow observables can always be constructed~\cite{Goldstein13,Malabarba14}. 

In this article we have found a set of sufficient conditions that ensure this.
More precisely, when the
distribution $p_\alpha \equiv \tfrac{1}{Q\|A\|} |\rho_{jk}| |A_{kj}|$ --which characterises the energy gaps that are most
relevant to the particular state and observable under consideration-- is approximately unimodal and spread over many different values, 
one expects $a(\epsilon) \sim 1$ and $\delta(\epsilon) \ll 1$. 
In the setting of a system interacting with a thermal bath, this implies equilibration time scales that do not scale as the size of the bath grows, as Theorem~\ref{theorem-boundmc} shows. 
 
Whether the above holds or not ultimately boils down to the values of the off-diagonal matrix elements of the observable and initial state in the energy basis, and to the
distribution of energy gaps.
Nevertheless, there are general arguments indicating that $a(\epsilon) \sim 1$ and $\delta(\epsilon) \ll 1$ might hold for a wide range of systems. 
Firstly, in typical situations one might expect that a state is spread roughly equally over a
range of energies (this occurs for example for thermal states), and that,  unless the
observable $A$ is fine-tuned,  its components are also spread relatively
smoothly over this band of energies, at least in a coarse-grained sense.
Secondly,
for large systems the distribution of energy levels tend to grow exponentially with energy in the region of finite temperature. It is easy to check that if one assumes a density of states $\nu(E) \propto e^{\beta E}$, the corresponding density of gaps scales like
$\mu\left(\left|G\right|\right) \propto e^{-\beta \left|G\right|} $,
with an exponential decrease.
This implies that, in order to have a resulting distribution that is characterised by one peak, it is \emph{sufficient} to have matrix elements of $A$ and of the initial state $\rho_0$ that grow sub-exponentially as a function of the energy gaps. 
This does not seem like a particularly strong assumption. 
Finally, let us note that even when $p_\alpha$ is not unimodal, and is instead composed of a number of distinct peaks, Theorem~\ref{theorem-boundmc} will still give reasonable equilibration times (in particular, times which are approximately independent of the size of the bath) except in the case where the individual peaks get sharper as the size of the bath is increased. 
Our intuition is that such behavior is rare, and thus that the bound will have very general applicability. 

The remaining question is whether physically relevant cases will in general be of this form (satisfying 
$a(\epsilon) \sim 1$ and $\delta(\epsilon) \ll 1$,
and therefore ``reasonably fast equilibrating'') or of the other (violating these conditions and therefore ``slow equilibrating''). 
Appendix~\ref{app:simulation} illustrates the transition to approximate unimodality, and the conditions being met, as environment size increases in a simulation of a 1D Ising model
with transversal magnetic field and periodic boundary conditions. However, proving that this occurs and finding the physical conditions under which it happens remains an interesting open problem for future study.

It is worthwhile comparing our conditions with the assumptions made in previous work in order to prove equilibration of closed quantum systems.
Equilibration can be proven by assuming that the effective dimension defined as $\de = \frac{1}{\sum_j \rho_{jj}^2}$ (for non degenerate spectrum for simplicity) is large~\cite{Reimann08,Lin09}, and that the Hamiltonian does not have too many degenerate energy gaps~\cite{ShortFarrelly11}. Notice that, although we do not make these assumptions explicitly, we are in some sense implicitly assuming both of them. On the one hand, a high effective dimension is related to having many energy levels populated in the system, which is necessary 
in order to have a distribution $p_\alpha$ that is spread over many different values.
On the other hand, the presence of a very degenerate energy gap results in a distribution $p_\alpha$ such that $\delta(\epsilon) \ll 1$ does not hold, as the simple example after Theorem~\ref{theorem-boundmc} illustrates.

The present paper emphasizes the importance of the off-diagonal matrix elements of observable and initial state to the study of the equilibration time scales in closed quantum systems. 

The Eigenstate Thermalization Hypothesis, introduced by Deutsch and Srednicki as a sufficient condition for thermalization~\cite{Deutsch91,Srednicki99}, has motivated extensive work on the distribution of diagonal matrix elements of observables~\cite{Rigol07,Rigol08,Rigol09,Rigol10,Beugeling14,Steinigeweg14,Khodja14,Kim14}.
However, much less work dealing with the distribution of the off-diagonal matrix elements is available, some examples being~\cite{FeingoldyPeres86,Steinigeweg13,Beugeling14offdiag,Fuchs14}. The recent papers \cite{Steinigeweg13,Beugeling14offdiag} show, in certain models, gaussian distribution of these matrix elements for local observables, which supports our claim that a unimodal distribution of $p_\alpha$ is to be expected in many situations. Moreover,~\cite{Fuchs14} numerically verifies our predictions in an experimentally realizable setup consisting of an electron in a quantum dot interacting with a bath of nuclear spins. 
Remarkably, the authors find that, even though our results are model independent and not tailored to this particular system, our new bounds fall within two orders of magnitude of the actual time scale.

Our paper focuses on the equilibration of a small system with respect to a 
pre-equilibrated bath, but many open questions remain regarding  general equilibration timescales. One direction of particular interest is an equilibration timescale for the bath itself, and what aspects are necessary for it to play its usual thermodynamic role. We hope that the tools developed here will aid in further study along these lines and help shed further light into this important topic.

\section*{Acknowledgements}

We would like to thank Moritz Fuchs, Daniel Hetterich, and Bj\"orn Trauzettel for many illuminating discussions, and for sharing the findings of~\cite{Fuchs14} prior to publication. 
Part of this work was supported by the COST Action
MP1209 ``Thermodynamics in the quantum regime''.
ASLM acknowledges support from the CNPq.
AJS acknowledges support from the Royal Society. 
AW is supported by the EU (STREP
``RAQUEL''), the ERC (AdG ``IRQUAT''), the Spanish
MINECO (grant FIS2013-40627-P) with the support of
FEDER funds, as well as by the Generalitat de Catalunya
CIRIT, project 2014-SGR-966.

\bibliography{references}

\appendix

\section{Equilibration in terms of the distinguishability}
\label{app:distinguishability}

Equilibration of the expectation value of some observable does not imply equilibration of the observable itself. Here we show how the results can be cast into a stronger sense of equilibration in terms of the distinguishability, as we used in~\cite{Malabarba14}. Distinguishability of states $\rho$ and $ \sigma$ with respect to an observable $\mathcal{M} = \left\{ P_1,P_2,\ldots,P_N \right\} $, where $P_i$ are a complete set of projectors, is defined by
\begin{align}
\mathcal{D}_{\mathcal{M}}(\rho,\sigma) = \frac{1}{2}\sum_i \big| \tr{P_i \rho} - \tr{P_i \sigma} \big|,
\end{align}
and it characterizes the probability of successfully guessing between the two states (assuming they are given with equal probabilities), via $p_{\text{success}} = \tfrac{1}{2} + \tfrac{1}{2}\mathcal{D}_{\mathcal{M}}(\rho,\sigma)$. By Jensen's and Cauchy-Schwarz's inequalities we can relate the distinguishability $\mathcal{D}_{\mathcal{M}}$ to the weak-distinguishability $\widetilde{\mathcal{D}}_{P_i}$ considered in this paper:
\begin{align}
&\big\langle \mathcal{D}_{\mathcal{M}}(\rho_t,\omega) \big\rangle_T \le \sqrt{\left\langle \mathcal{D}_{\mathcal{M}}^2(\rho_t,\omega) \right\rangle_T} \nonumber \\
& \qquad \le  \sqrt{N} \sqrt{\frac{1}{4}\sum_i \left\langle \big| \tr{P_i \rho_t} - \tr{P_i \omega} \big|^2 \right\rangle_T}  \nonumber \\
& \qquad =  \sqrt{N} \sqrt{\sum_i \left\langle \widetilde{\mathcal{D}}_{P_i}(\rho_t,\omega) \right\rangle_T}.
\end{align}
Each term $\left\langle \widetilde{\mathcal{D}}_{P_i}(\rho_t,\omega) \right\rangle_T$ can be bounded via the results from Theorems~\ref{theorem-genboundtime} and \ref{theorem-boundmc}, and therefore fast equilibration of the projectors will imply fast equilibration of the distinguishability.

\section{The $\xi_p \big( \frac{1}{T} \big)$ function -- Proposition~\ref{proposition-xibound}}
\label{app:xifunc}

We wish to bound the decay with time $T$ of
\begin{equation}
\sum_{\alpha \beta  } p_\alpha p_\beta e^{-|G_\alpha -G_\beta| T}.
\end{equation}

To connect this to the function
\begin{align}
   \xi_p \big( \tfrac{1}{T} \big) = \max_{G \in \mathbb{R}} \sum_{\substack{\alpha \;: \\ G_\alpha \in \big[ G, G + \tfrac{1}{T}\big] }}  p_\alpha,
\end{align}
we define the auxiliary function
\begin{equation}
    \label{eq:16}
    g(x) =
    \begin{cases}
        1, \mbox{ if } x \in [0,1) \\
        0, \mbox{ otherwise. }
    \end{cases}
\end{equation}
This definition allows to upper bound the exponential as
\begin{equation}
    \expo{-\left| x \right|}\leq \sum_{n = 0}^{\infty} \expo{-n} g \left( \left| x \right| - n \right).
\end{equation}
One can then see
\begin{align}
    \label{eq:4}
    &\!\!\!\!\!\!\!\!\sum_{\alpha \beta}p_{\alpha}p_{\beta} e^{- |G_\alpha-G_\beta |T } \nonumber \\
    &\leq \sum_{n = 0}^{\infty} \expo{-n} \sum_{\alpha} p_\alpha \sum_{\beta} p_\beta g \left( \left| G_\alpha - G_\beta \right|T - n \right) \notag\\
    &= \sum_{n = 0}^{\infty} \expo{-n} \sum_{\alpha} p_\alpha \!\! \!\!
    \sum_{\substack{\beta\;: \\ \left(\left| G_\alpha - G_\beta \right| T - n \right) \in [0, 1)}} \!\! \!\! \!\! \!\! p_\beta \notag\\
    & \leq \sum_{n = 0}^{\infty} \expo{-n} \sum_{\alpha} p_\alpha
    \left[\sum_{\substack{\beta\;:\\ G_\beta \in I_{-}^{(n)}}} p_\beta + \sum_{\substack{\beta\;:\\ G_\beta \in I_{+}^{(n)}}} p_\beta \right]
    \notag\\
    &\leq \sum_{n = 0}^{\infty} \expo{-n} \sum_{\alpha} p_\alpha \Big[ 2 \xi_p \big( \tfrac{1}{T} \big) \Big] \notag\\
    &= \frac{2}{1-\expo{-1}} \xi_p \big( \tfrac{1}{T} \big)
\end{align}
where $I^{(n)}_{+}\! =\! \left[G_\alpha + \frac{n}{T}, G_\alpha+ \frac{(n+1)}{T} \right)$ and  
$I_{-}^{(n)} \!=\! \left(G_{\alpha} - \frac{(n+1)}{T}, G_\alpha - \frac{n}{T}\right]$,
 and the inequality in the
penultimate line is valid for any $n$ and $\alpha$.

Together with equation~(\ref{eqapp:xi}),
\begin{align}
\left\langle  \mathcal{\widetilde D}_A(\rho_t,\omega) \right\rangle_T \le \frac{5 \pi Q^2}{16} \sum_{ \alpha \beta   }   p_\alpha p_\beta e^{-|G_\alpha -G_\beta| T},
\end{align}
we obtain
\begin{align}
        \left\langle  \mathcal{\widetilde D}_A(\rho_t,\omega) \right\rangle_T &\le \frac{5 \pi }{8(1-e^{-1})} Q^2 \xi_p \big( \tfrac{1}{T} \big) \nonumber \\
        &\le \pi Q^2 \xi_p \big( \tfrac{1}{T} \big),
    \end{align}
    showing Proposition~\ref{proposition-xibound}.

\section{Truncation of the Hilbert space -- Proposition~\ref{proposition-truncationmc}
}
\label{app:truncation-errorsmc}

As mentioned in the main text, the state $\rho_0 = \rho_S \otimes \frac{\id_B^\Delta}{d_B^\Delta} $ lies within the subspace $\mathcal{H}_{H_S+H_B}^{E_B,\Delta+2\|H_S\|}$.
We will show that, when considering the full interacting Hamiltonian $H = H_S+H_B+H_I$, it is enough to consider the truncated subspace
$\mathcal{H}_{H_S+H_B+H_I}^{E_B,\Delta + 2\|H_S\| + \eta}$, where $\eta$ denotes the amount by which the energy window is extended.

The effect of ``cutting'' the state outside the space $\mathcal{H}_{H_S+H_B+H_I}^{E_B,\Delta + 2\|H_S\| + \eta}$ will obviously introduce errors. We wish to do it such that the truncated state remains close to the original one in trace distance:
\begin{align}
    \| \rho_0 - \Pi \rho_0 \Pi \|_1 & \le  \frac{2}{K}, 
\end{align}
    where $\Pi$ is the projector onto $\mathcal{H}^{E_B,\Delta + 2 \| H_S \| + \eta}_{H_B + H_S + H_I}$ and $\eta = \sqrt{8 d_S} \|H_I\| K$. 

An arbitrary initial state of the system can be written as
\begin{equation}
\rho_S = \sum_{l,l'} \rho_{ll'}^S \ket{E_l^S} \bra{E_{l'}^S},
\end{equation}
where  $\ket{E_l^S} $ are eigenvectors of $H_S$ with energy $E_l^S$.
As seen in the main text, the state of the bath is effectively proportional to the identity onto a window of $H_B$, and therefore diagonal in the basis $H_B$.
Hence, we can write the total state as
\begin{align}
\rho_0 &= \sum_\lambda \sum_{l,l'} c_\lambda \rho_{ll'}^S \ket{E_l^S} \bra{E_{l'}^S} \otimes \ket{E_\lambda^B} \bra{E_\lambda^B} \nonumber \\
&\equiv \sum_{\lambda,l,l'} c_\lambda \rho_{ll'}^S \ket{E_l^S, E_\lambda^B} \bra{E_{l'}^S, E_\lambda^B},
\end{align}
where $\ket{E_\lambda^B}$ are the eigenvectors of $H_B$ and the coefficients $c_\lambda$ are positive and normalized (we will do the calculation for an arbitrary state of the bath commuting with $H_B$, but in our case actually $c_\lambda = 1/d_B^\Delta$).

The following result will be useful.
\begin{lemma}[Gentle measurement~\cite{AWgentlemeasurement}]
For any state $\rho$ and positive operator $X$ such that $X \le I$ and 
\begin{equation}
1-\tr{\rho X} \le \frac{1}{2K^2} \le 1, \end{equation}
one has
\begin{equation}
\left\| \rho- \sqrt{X} \rho \sqrt{X} \right\|_1 \le \frac{2}{K}.
\end{equation}
\end{lemma}

In order to apply the lemma for the state $\rho = \rho_0$ and the operator $X = \sqrt{X} = \Pi$, we start with
\begin{align}
\label{eq:apptruncationAUX3}
&1-\tr{\rho_0 \Pi} = \tr{\rho_0 \Pi^\bot}  \\
&\qquad= \sum_{\lambda,l,l'} c_\lambda \rho_{ll'}^S \tr{ \ket{E_l^S, E_\lambda^B} \bra{E_{l'}^S, E_\lambda^B} \Pi^\bot} \nonumber \\
&\qquad \le \sum_{\lambda,l,l'} c_\lambda \left| \rho_{ll'}^S \right|   \left|  \bra{E_{l'}^S, E_\lambda^B} \Pi^\bot \ket{E_l^S, E_\lambda^B} \right| \nonumber \\
&\qquad\le \sum_{\lambda,l,l'} c_\lambda  \left| \rho_{ll'}^S \right|  \left\| \Pi^\bot \ket{E_l^S, E_\lambda^B} \right\|  \left\|  \Pi^\bot \ket{E_{l'}^S, E_\lambda^B} \right\|, \nonumber 
\end{align}
by taking absolute values in line 3, in line 4 using the Cauchy-Schwarz inequality (with $\left\| \ket{\psi}\right\| = \sqrt{\langle \psi| \psi \rangle}$ as usual), and denoting the orthogonal complement of $\Pi$ by $\Pi^\bot$.

In order to upper bound this expression, we use Bhatia's perturbation theory result (Theorem VII.3.1) in~\cite{Bhatia1997}  (for another very interesting application  treating the problem of proving thermalization in closed quantum systems see~\cite{Riera12}).
\begin{theorem}
\label{bhatiapert}
Let $O$ and  $P$ be normal operators,
$S_1$ and $S_2$ be two subsets of the complex plane that are separated by a strip (or annulus) of
width $\Delta$, and let $E$ ($F$) denote the orthogonal projection onto the subspace spanned by the eigenvectors of $O$ ($P$) corresponding to those
of its eigenvalues that lie in $S_1$ ($S_2$). Then, for every unitarily
invariant norm $\vert\kern-0.20ex\vert\kern-0.20ex\vert \cdot \vert\kern-0.20ex\vert\kern-0.20ex\vert$,
\begin{align}
\vert\kern-0.20ex\vert\kern-0.20ex\vert E F \vert\kern-0.20ex\vert\kern-0.20ex\vert \le \frac{ \vert\kern-0.20ex\vert\kern-0.20ex\vert O-P \vert\kern-0.20ex\vert\kern-0.20ex\vert }{\Delta}.
\end{align}
\end{theorem}

In our notation, this theorem 
implies 
\begin{equation} \label{eq:bhati} 
 \left\| \Pi^\bot \ket{E_l^S, E_\lambda^B} \right\| = \left|\left| \Pi^\bot  \ket{E_l^S, E_\lambda^B }\bra{E_l^S, E_\lambda^B} \right|\right| \le \frac{\|H_I\|}{\Delta_{l,\lambda}},
\end{equation}
where we have related the euclidean vector norm on the left to the operator norm on the right.  In the above expression $\Delta_{l,\lambda}$ is the distance between the supports of $\ket{E_l^S, E_\lambda^B }\bra{E_l^S, E_\lambda^B}$ and $\Pi^\bot$. Note that this distance satisfies  $\Delta_{l,\lambda}~\ge~\frac{\eta}{2}$.

Using equations~(\ref{eq:apptruncationAUX3}) and~(\ref{eq:bhati}), the fact that 
\begin{equation} 
\Bigg( \sum_{l,l'} |\rho_{ll'}^S |\Bigg)^2 \leq d_S^2 \sum_{l,l'} |\rho_{ll'}^S|^2 = d_S^2 \tr{\rho_S^2} \leq d_S^2,
\end{equation} and $\sum_\lambda c_{\lambda} =1$,   we get  
\begin{align}
1-\tr{\rho_0 \Pi} &\le \sum_{\lambda,l,l'}  c_\lambda  \left| \rho_{ll'}^S \right| \frac{4\|H_I\|^2}{\eta^2} \nonumber \\
&\le d_S \frac{4\|H_I\|^2}{\eta^2}.
\end{align}

Choosing the truncation window with $\eta~=~K \sqrt{8 d_S }\|H_I\| $ and using the gentle measurement lemma leads to the main result
\begin{align}
\label{eq3:appaux7}
\left\| \rho_0 - \Pi \rho_0 \Pi  \right\|_1 & \le  \frac{2}{K}.
\end{align}

It is easy to extend this to the weak-distinguishability.
Note that the trace distance is invariant under global unitaries, in particular invariant under the Hamiltonian evolution. Therefore, expectation values of an observable $A$ will be close, since we have
\begin{align}
\left| \tr{A \rho_t } - \tr{A \Pi \rho_t  \Pi} \right| &\le \| \rho_t  - \Pi \rho_t  \Pi \|_1 \|A\| \nonumber \\
&= \| \rho_0 - \Pi \rho_0 \Pi \|_1 \|A\| \nonumber \\
&\le \frac{2\|A\|}{K},
\end{align}
by using in the first step that, for any two operators $O$ and $P$, Holder's inequality implies $\tr{OP} \le \| O P \|_1 \le \| O \|_1 \| P \|$.

From above we obtain
\begin{align}
\label{eq:app:auxdistance}
\mathcal{\widetilde D}_A \left( \rho_t , \Pi \rho_t  \Pi \right) &\le \frac{1}{K^2}.
\end{align}

To conclude, it will be useful to relate the rank of $\Pi$ to the rank of a projector $P$ onto yet another extended subspace $\mathcal{H}^{E_B,\Delta + 2 \| H_S \| + \eta + \eta'}_{H_B + H_S}$. Since  the rank of $P$ is related to the spectrum of unperturbed system and bath Hamiltonians it will prove simpler to calculate. 

First, denoting by $P^\bot$ the orthogonal projector to $P$ we get 
\begin{align}
\| \Pi - \Pi P\|_1 &= \| \Pi P^\bot \|_1 \nonumber \\
&\le \text{rank} \left( \Pi \right) \| \Pi P^\bot \| \nonumber \\
&\le \| \Pi \|_1 \frac{2 \|H_I\|}{ \eta' }  = \frac{\| \Pi \|_1}{K},
\end{align}
by using the fact that for any operator $\| Q \|_1 \le \text{rank} \left( Q \right) \|Q \|$, Bhatia's theorem, and setting $\eta' = 2 K \| H_I \|$. 
The triangle inequality then leads to
\begin{align}
\| \Pi \|_1 &\le \| \Pi P \|_1 +  \| \Pi - \Pi P \|_1 \nonumber \\
&\le \| P \|_1 + \frac{\| \Pi \|_1}{K}.
\end{align}  
Recall that $d_\text{trunc} = \text{rank}\left( \Pi \right) = \| \Pi \|_1$. Hence
\begin{align}
d_\text{trunc} \le \frac{\| P \|_1}{1-\tfrac{1}{K}}.
\end{align}

Note that, since $P$ projects onto the subspace $\mathcal{H}^{E_B,\Delta + 2 \| H_S \| + \eta + \eta'}_{H_B + H_S}$ corresponding to a system that does not interact with the bath, we can denote the density of states of the bath by $\nu_B(E)$, and
\begin{align}
\| P \|_1 &\le d_S \int_{E_B - \tfrac{\Delta}{2} - \tfrac{\eta}{2} - \tfrac{\eta'}{2} - \|H_S\|  }^{E_B + \tfrac{\Delta}{2} + \tfrac{\eta}{2} + \tfrac{\eta'}{2} + \|H_S\|  } \nu_B(E) dE  \\
& = d_S \int_{E_B - \tfrac{\Delta}{2} -  \left( 1 + \sqrt{2d_s} \right) K \|H_I\| -\|H_S\|  }^{E_B + \tfrac{\Delta}{2} +\left( 1 + \sqrt{2d_s} \right) K \|H_I\| + \|H_S\|  } \nu_B(E) dE. \nonumber
\end{align}
The inequality comes from upper bounding the number of accessible states of the system by $d_S$, and counting the states of the bath as if it could access all of the possible energies of the space $\mathcal{H}^{E_B,\Delta + 2\| H_S \| + \eta + \eta'}_{H_B + H_S}$. The second line comes from using $\eta' = 2 K \| H_I \|$ and $\eta = \sqrt{8d_S} K \| H_I \|$.

\section{Proof details of Theorem~\ref{theorem-boundmc}}
\label{app:truncation-commutators}

We write the original state in the basis $\ket{j}$ of eigenvectors of the full Hamiltonian $H$ as
\begin{equation}
\rho_0 = \sum_{jk} \rho_{jk}  \ket{j} \bra{k}.
\end{equation}
The truncated state is then
\begin{equation}
\Pi \rho_0 \Pi = \sum_{jk \in J} \rho_{jk}  \ket{j} \bra{k},
\end{equation}
where $\Pi$ projects to the truncated Hilbert space $\mathcal{H}_{H_S+H_B+H_I}^{E_B,\Delta + 2 \|H_S\| + \eta}$, and $J$ is the set of eigenvalues of $H_S+H_B+H_I$ restricted to such space.

We now expand the Hilbert space to two include two new energy eigenvectors $\ket{j_{\min}}$ and $\ket{j_{\max}}$, with corresponding energies $E_{\min}$ and $E_{\max}$ respectively,  such that the new Hamiltonian is 
\begin{equation} 
\overline{H} \equiv H + E_{\min} \proj{j_{\min}} + E_{\max} \proj{j_{\max}}. 
\end{equation}
Next, we define a new density matrix on the enlarged space by 
\begin{align}
\overline{\rho}_0 &\equiv \zeta \left( \Pi \rho_0 \Pi + \frac{|x|}{2} \ket{j_{\min}} \bra{j_{\min}} + \frac{x}{2} \ket{j_{\min}} \bra{j_{\max}} \right. \nonumber \\ &\qquad \left. \qquad+ \frac{x}{2} \ket{j_{\max}} \bra{j_{\min}} + \frac{|x|}{2} \ket{j_{\max}} \bra{j_{\max}} \right),
\end{align}
where $x$ is a real constant, and $\zeta$ is an appropriate normalization constant to ensure that $\tr{\overline{\rho}_0}=1$. The above definition ensures that $\overline{\rho}_0$ remains a positive  operator.
Note also that $\Pi$ is orthogonal to $\ket{j_{\min}}$ and $\ket{j_{\max}}$, since the truncation we perform is in the original Hilbert space.

It will also be useful to define an observable $\overline{A}$ in a similar way, as
\begin{align}
\overline{A}  &\equiv \frac{1}{\zeta} \left( \Pi A \Pi + \frac{\|A\|}{2} \ket{j_{\min}} \bra{j_{\min}} + \frac{\|A\|}{2} \ket{j_{\min}} \bra{j_{\max}} \right. \nonumber \\ &\qquad \left. \quad+ \frac{\|A\|}{2} \ket{j_{\max}} \bra{j_{\min}} + \frac{\|A\|}{2} \ket{j_{\max}} \bra{j_{\max}} \right).
\end{align}

In the derivation of Theorem~\ref{theorem-genboundtime}, the commutators which appear in the denominator of equation~(\ref{eq3:genboundtime}) came from the standard deviation of the distribution $p_\alpha$, as explained in 
Section~\ref{sec:obsdepbound}. 
We intend to use $\overline{\rho}_0$ and $\overline{A}$ as the initial state and observable for our calculation, while proving that $x$ can be taken such that:
\begin{enumerate}[label=(\roman*)]

\item  
\label{auxstatespoint1}
$\qquad \overline{\sigma}_G^2 \ge \frac{1}{\overline{Q} \|\overline{A}\| }\Big|   \text{Tr}\Big(\big[ \left[\rho_0  ,H \right] , H\big] A  \Big) \Big|,$

\item
\label{auxstatespoint2}
$\left\{ \overline{\rho}_0,\overline{A}, \overline{H} \right\}$ lead to approximately the  same physics as $\left\{ \rho_0,A, H\right\}$, 

\end{enumerate}
where $\overline{Q}$ is the normalization factor for the distribution $\overline{p}_\alpha$ corresponding to $\overline{\rho}_0$ and $\overline{A}$.

Point \ref{auxstatespoint1} will lead to a result in Theorem \ref{theorem-boundmc} with commutators involving the original triple $\left\{ \rho_0,A, H\right\}$, as desired, 
while \ref{auxstatespoint2} allows to  use these redefined state, observable and Hamiltonian instead of the original ones.

\subsection{Going from $\left\{ \overline{\rho}_0,\overline{A}, \overline{H} \right\}$ to commutators involving $\left\{ \rho_0,A, H\right\}$ \label{app:truncation-commutators-1} }

Let us revisit equation~(\ref{eq:app-sigma}), this time with the state $\overline{\rho}_0$ and the observable $\overline{A}$. We can see
\begin{align}
\label{eq:app-overlinesigma1}
&\overline{\sigma}_G^2 = \sum_\alpha \overline{p_\alpha} G_\alpha^2  = \frac{1}{\overline{Q} \|\overline{A}\|} \sum_{j \ne k } |\overline\rho_{jk} | |\overline A_{kj}| (E_j-E_k)^2 \nonumber \\
&= \frac{1}{\overline{Q}\|\overline{A}\|}\left( \sum_{jk \in J} |\rho_{jk} | |A_{kj}| (E_j-E_k)^2 + 2 \frac{|x|}{2} \frac{\|A\|}{2} G_{\max}^2 \right) \nonumber \\
&\ge \frac{1}{\overline{Q} \| \overline{A} \|} \left( \Big|   \sum_{jk \in J} \rho_{jk}   A_{kj} (E_j-E_k)^2   \Big|   +  \frac{x \|A\|}{2} G_{\max}^2 \right), 
\end{align}
where we defined the maximum energy gap $G_{\max} = E_{\max} - E_{\min}$.

We now \emph{impose} that $x$ is such that
\begin{align}
\label{eq:app-overlinesigma2}
&\frac{1}{\overline{Q} \| \overline{A} \|} \left( \Big|   \sum_{jk \in J} \rho_{jk}   A_{kj} (E_j-E_k)^2   \Big|   +  \frac{x \|A\| }{2} G_{\max}^2 \right) \nonumber \\
&\qquad \equiv \frac{1}{\overline{Q} \|\overline{A}\| } \Big|   \sum_{jk} \rho_{jk}   A_{kj} (E_j-E_k)^2   \Big| \nonumber \\
&\qquad = \frac{1}{\overline{Q} \|\overline{A}\| }\Big|   \text{Tr}\Big(\big[ \left[\rho_0  ,H \right] , H\big] A  \Big) \Big|,
\end{align}
which together with eq.~(\ref{eq:app-overlinesigma1}) already gives condition \ref{auxstatespoint1}.
 
From the equation above we obtain
\begin{align}
& x \frac{\|A\| G_{\max}^2}{2}  = \nonumber \\
& \quad \Big|   \sum_{jk} \rho_{jk}   A_{kj} (E_j-E_k)^2   \Big| -  \Big|   \sum_{jk \in J} \rho_{jk}   A_{kj} (E_j-E_k)^2   \Big| \nonumber \\
& \le \Big|   \sum_{jk} \rho_{jk}   A_{kj} (E_j-E_k)^2   -  \sum_{jk \in J} \rho_{jk}   A_{kj} (E_j-E_k)^2   \Big| \nonumber \\
& \qquad  = \Big|   \text{Tr}\Big(\big[ \left[\rho_0 - \Pi \rho_0 \Pi , H \right] , H\big] A  \Big) \Big| \nonumber \\
& \qquad  = \Big|   \text{Tr}\Big(\big[ \left[\rho_0 - \Pi \rho_0 \Pi , H' \right] , H'\big] A  \Big) \Big|,
\end{align}
where in the last line we define the auxiliary Hamiltonian $H' \equiv H - \tfrac{E_{\max} + E_{\min}}{2} \id$, shifted so that its spectrum is centered around $0$.
By noting that $\|H'\| = \|H\| - \tfrac{E_{\max} + E_{\min}}{2} = \tfrac{E_{\max} - E_{\min}}{2} = \tfrac{G_{\max}}{2}$, expanding the four terms in the above commutators, using that for any two operators $O$ and $P$ one has $\tr{O P} \le \|O\| \|P\|_1$ from Holder's inequality, and equation~(\ref{eq3:appaux7}), we end up with
\begin{align}
x \  &\le \frac{2}{\|A\| G_{\max}^2} 4 \| \rho_0 - \Pi \rho_0 \Pi \|_1 \|A\|  \|H'\|^2 \nonumber \\
& = 2 \| \rho_0 - \Pi \rho_0 \Pi \|_1 \nonumber \\
& \le \frac{4}{K}.
\end{align}
Similarly, we can show that $-x \leq \frac{4}{K}$, and hence that 
\begin{equation}
|x| \leq \frac{4}{K}.
\end{equation}

\subsection{$\left\{ \overline{\rho}_0,\overline{A}, \overline{H} \right\}$ give approximately the \protect\\ same physics as $\left\{ \rho_0,A, H\right\}$}
\label{app:truncationmc-redefinedstates}

We now check that condition \ref{auxstatespoint2} is also satisfied.
Note that
\begin{align}
&\big| \tr{\rho_0 A} - \tr{ \overline{\rho}_0 \overline{A}} \big|  \le \nonumber \\
&\big| \tr{\rho_0 A} - \tr{ \Pi \rho_0 \Pi A} \big| + \big| \tr{ \Pi \rho_0 \Pi A} - \tr{ \overline{\rho}_0 \overline{A} } \big| \nonumber \\
&= \big| \tr{\rho_0 A} - \tr{ \Pi \rho_0 \Pi A} \big| + \left|  \frac{x \|A\|}{2} +  \frac{|x|    \|A\|}{2} \right| \nonumber \\
& \qquad \quad\le \|  \rho_0 - \Pi \rho_0 \Pi \|_1   \|A\| + |x|  \|A\| \nonumber \\
& \qquad \quad \le \frac{2 \|A\|}{K} + \frac{4 \|A\|}{K}, 
\end{align}
by using in the third line that $\tr{\overline{\rho}_0 \overline{A}} = \tr{A \Pi \rho_0 \Pi} + 2 \frac{x \|A\|}{4} + 2 \frac{|x|  \|A\|}{4} $ (which comes from the definition of $\overline{A}$ and $\overline{\rho}_0$), and the fact that $\Pi$ is orthogonal to $\ket{j_{\min}}$ and $\ket{j_{\max}}$.

The above equation justifies the approach of defining auxiliary state and observable, since we proved that these mimic the original state and observable in the predictions.

The result can be translated into the weak-distinguishability between $\rho_t $ and $\omega $. By the triangle inequality, a similar calculation as above, and the fact that the trace distance is invariant under unitary evolution, we see
\begin{align}
& \left| \tr{\rho_t  A} - \tr{\omega  A} \right|^2 \le \left| \tr{\rho_t  A} - \tr{\overline{\rho}_t  \overline{A}} \right|^2 \nonumber \\
&\qquad \qquad \qquad \qquad + \left| \tr{\overline{\rho}_t  \overline{A}} - \tr{\overline{\omega}  \overline{A}} \right|^2 \nonumber \\
&\qquad \qquad \qquad \qquad + \left| \tr{\overline{\omega}  \overline{A}} - \tr{\omega  A} \right|^2 \nonumber \\
& \qquad \le \left| \tr{\overline{\rho}_t  \overline{A}} - \tr{\overline{\omega}  \overline{A}} \right|^2 + 2\left(\frac{6 \|A\|}{K} \right)^2,
\end{align}
where $\overline{\rho}_t  = e^{-i \overline{H}t} \rho_t  e^{i \overline{H}t}$ and $\overline{\omega} $ is the corresponding dephased state. This implies
\begin{align}
 \label{eq:app:auxdistingmc1}
\mathcal{\widetilde D}_A \left( \rho_t , \omega  \right) &\le \frac{\| \overline{A} \|^2}{\| A \|^2} \mathcal{\widetilde D}_{\overline{A}} \left( \overline{\rho}_t , \overline{\omega}  \right) + \frac{18}{ K^2}.
\end{align}

For the term $\mathcal{\widetilde D}_{\overline{A}} \left( \overline{\rho}_t , \overline{\omega}  \right)$ we can apply Propositions~\ref{proposition-xibound} and~\ref{proposition-xidependence}, and~condition\ref{auxstatespoint1} to get 
\begin{align}
        \label{eq:app:auxdistingmc2}
      & \left\langle \mathcal{\widetilde D}_{\overline A} \left( \overline{\rho}_t  , \overline{\omega}  \right) \right\rangle_T \le \frac{\pi\, a(\epsilon) \overline Q^{2}}{T \overline{\sigma}_G} + \pi\, \delta(\epsilon) \overline Q^2  \\
        & \qquad \qquad \le \frac{\pi \| \overline{A} \|^{1/2} a(\epsilon) \overline Q^{5/2}}{T \sqrt{\Big| \text{Tr}\Big( \big[ \left[ \rho_0 ,H \right] , H\big] A \Big) \Big|}} + \pi\, \delta(\epsilon) \overline Q^2. \nonumber
    \end{align}
The last inequality comes from our convenient construction of $\overline{\rho}_0$ and $\overline{A}$, specifically designed for this.

\subsection{The factor $\overline{Q}$ for $\overline{\rho}_0$ and $\overline{A}$}
\label{app:truncationmc-redefinedstatesQfactor}

It is easy to see that the factor $\overline Q$ for the auxiliary state and observable satisfies
\begin{align}
\overline Q &\equiv \sum_{j \ne k } |\overline\rho_{jk} | \frac{|\overline A_{kj}|}{\|\overline{A}\|} \nonumber \\
& = \sum_{jk \in J} |\rho_{jk} | \frac{|A_{kj}|}{\|\overline{A}\|} + 2\frac{|x|}{2}\frac{\|A\|}{2\|\overline{A}\|} \nonumber \\
&  = \frac{\| A \|}{\| \overline{A} \|}\left( Q_\text{trunc} + \frac{|x|}{2} \right) \nonumber \\
&  \le \frac{\| A \|}{\| \overline{A} \|}\left( Q_\text{trunc} + \frac{2}{K} \right),
\end{align}
where $Q_{\text{trunc}}$ is the normalization constant of the distribution that results from $\Pi \rho_0 \Pi$ and $A$.

The above bound, plus equations~(\ref{eq:app:auxdistingmc1}) and~(\ref{eq:app:auxdistingmc2}), results in
\begin{align}
\left\langle \mathcal{\widetilde D}_A \left( \rho_t , \omega  \right) \right\rangle_T &\le \frac{\| \overline{A} \|^2}{\| A \|^2} \frac{\pi\, a(\epsilon)  \| \overline{A} \|^{1/2} \overline Q^{5/2}}{T \sqrt{\Big| \text{Tr}\Big( \big[ \left[ \rho_0 ,H \right] , H\big] A \Big) \Big|}} \nonumber \\
&  + \frac{\| \overline{A} \|^2}{\| A \|^2}  \pi\, \delta(\epsilon) \overline Q^2  + \frac{18}{ K^2} \nonumber \\
& \le \frac{\pi\, a(\epsilon)  \| A \|^{1/2} \left( Q_\text{trunc} + \tfrac{2}{K} \right)^{5/2}}{T \sqrt{\Big| \text{Tr}\Big( \big[ \left[ \rho_0 ,H \right] , H\big] A \Big) \Big|}} \nonumber \\
&  +  \pi\, \delta(\epsilon) \left( Q_\text{trunc} + \tfrac{2}{K} \right)^2  + \frac{18}{ K^2},
\end{align}
which, defining $Q_2 = Q_\text{trunc} + \tfrac{2}{K}$ to simplify notation, 
gives the first part of
Theorem~\ref{theorem-boundmc}.

In order to finish the Theorem's proof we upper bound $Q_\text{trunc}$. From equation~(\ref{eq:app-Qbound}) we see that for this state
\begin{align}
Q_\text{trunc}^2 &\le \tr{ \left( \Pi \rho_0 \Pi \right)^2 } \tr{\Pi}  \nonumber\\
&= \tr{  \Pi \rho_0 \Pi \rho_0 }  d_{\text{trunc}} \nonumber \\
&\le \sqrt{\tr{  \Pi  \rho_0^2 \Pi } \tr{ \rho_0\Pi^2 \rho_0  }} d_{\text{trunc}} \nonumber \\
&\le \|\Pi\| \tr{ \rho_0^2} d_\text{trunc}  \nonumber \\
&=\tr{ \rho_0^2} d_\text{trunc} ,
\end{align}
by using the definition of $d_\text{trunc}$, the Cauchy-Schwarz inequality, and the fact that for two positive semidefinite matrices $\tr{OP} \le \|O\| \tr{P}$.
  
Since $\rho_0 = \rho_S\otimes \frac{\id_B^\Delta}{d_B^\Delta}$ we see
\begin{align}
\tr{ \rho_0^2} = \frac{\trs{\rho_S^2}}{d_B^\Delta} .
\end{align}

In the main text we found (eq.~(\ref{eq:truncateddimension}))
\begin{equation}
d_\text{trunc} \le\frac{d_S d_B^\Delta}{\left(1-\tfrac{1}{K}\right)\left(1-e^{-\beta \Delta} \right)} e^{ \beta \|H_S\| +  \big( 1 + \sqrt{2d_s} \big) K \beta \|H_I\| } ,
\end{equation}
which leads to
\begin{align}
&Q_\text{trunc}^2 \le \tr{ \rho_0^2} d_\text{trunc}  \\
& \le   \frac{d_S \trs{\rho_S^2}}{\left(1-\tfrac{1}{K}\right)\left(1-e^{-\beta \Delta} \right)} e^{ \beta \|H_S\| +  \big( 1 + \sqrt{2d_s} \big) K \beta \|H_I\| } . \nonumber
\end{align}
Substituting this back into the definition for $Q_2$ gives eq.~(\ref{eq:boundmc-Qfactor}).

\section{Typical behaviour for environment in a pure initial state -- Proposition~\ref{proposition-typicalbath}}
\label{app3:typical-bath}

With the initial state
\begin{equation}
\rho_0^U = \rho_S \otimes U \ket{\psi} \bra{\psi} U^\dag = \rho_S \otimes \rho_B^U
\end{equation}
we focus on the average over all unitaries $U$ within the subspace $ \mathcal{H}_{H_B}^{E_B,\Delta} $of $\mathcal{\widetilde D}_A\big(\rho_t^U,\omega^U\big)$:
\begin{equation}
\big\langle \mathcal{\widetilde D}_A\big(\rho_t^U,\omega^U\big) \big\rangle_U = \frac{\Big\langle \Big| \tr{\rho_t^U A} - \tr{\omega^U A} \Big|^2   \Big\rangle_U}{4\|A\|^2}.
\end{equation}

By shifting the time dependencies to the observable $A$ we get
\begin{equation}
\tr{\rho_t^U A} = \tr{\rho_0^U e^{i H t} A e^{-i H t}} \equiv \tr{\rho_0^U A(t)}.
\end{equation}
If $A_{eq}$ is the infinite time averaged $A(t)$, we have $\tr{\omega^U A} = \tr{\rho_0^U A_{eq}}$.
Then
\begin{align}
&\big\langle \mathcal{\widetilde D}_A\big(\rho_t^U,\omega^U\big) \big\rangle_U = \frac{1}{4 \|A\|^2} \left\langle \left| \tr{\rho_0^U \left( A(t) - A_{eq} \right) } \right|^2   \right\rangle_U   \nonumber \\
&\qquad= \frac{1}{4 \|A\|^2} \Big\langle \Big| \trb{\rho_B^U \trs{\rho_S \otimes \id_B \big( A(t)-A_{eq} \big)} } \Big|^2   \Big\rangle_U \nonumber \\
&\qquad= \Big\langle \Big| \trb{\rho_B^U C} \Big|^2 \Big\rangle_U  \nonumber \\
&\qquad= \text{Tr}_{B^{\otimes2 }} \big[\big\langle \rho_B^U\otimes \rho_B^U \big\rangle_U C \otimes C \big],
\end{align}
where in the third line we have used $\rho_B^U=\id_B^\Delta \rho_B^U\id_B^\Delta$ to write $C \equiv \frac{1}{2\|A\|} \trs{\rho_S \otimes \id_B^\Delta \big( A(t)-A_{eq} \big) \id_S \otimes \id_B^\Delta }$ as an observable acting on the microcanonical window of the bath Hilbert space.

Via the same calculations as used in the Appendix of~\cite{Malabarba14}, we get
\begin{align}
\big\langle \rho_B^U\otimes \rho_B^U \big\rangle_U &= \Big\langle U^{\otimes2} \left( \ket{\psi}\bra{\psi}  \otimes \ket{\psi}\bra{\psi} \right) (U^{\otimes2})^\dag \Big\rangle_U   \nonumber \\
&=\alpha \Pi_s + \beta \Pi_a,
\end{align}
where $\Pi_s = \frac{(\id_B^\Delta)^{\otimes2} + \SWAP}{2}$ and $\Pi_a = \frac{(\id_B^\Delta)^{\otimes2} - \SWAP}{2}$ project onto the symmetric and antisymmetric subspaces, and $\SWAP$ is the swap operator on $\mathcal{H}_{H_B}^{E_B,\Delta} \otimes \mathcal{H}_{H_B}^{E_B,\Delta} $, defined by $\SWAP\, \ket{\phi_1} \ket{\phi_2} = \ket{\phi_2} \ket{\phi_1} $. Since 
\begin{align} 
\tr{\big\langle \rho_B^U\otimes \rho_B^U \big\rangle_U \Pi_a} &= \frac{1}{2}\big\langle \left( 1 -\tr{\rho_B^U\otimes \rho_B^U \SWAP} \right) \big\rangle_U \nonumber \\
 &= \frac{1}{2}\big\langle \left( 1 - \tr{(\rho_B^U)^2} \right) \big\rangle_U \nonumber \\
&= 0
\end{align} 
we see that $\beta=0$, and from $\tr{\big\langle \rho_B^U\otimes \rho_B^U \big\rangle_U} =1$  we obtain  $\alpha = \frac{2}{d_B^\Delta(d_B^\Delta+1)}$, which leads to the simple expression $\big\langle \rho_0^U\otimes \rho_0^U \big\rangle_U  = \frac{2}{d_B^\Delta(d_B^\Delta+1)}\Pi_s $.
Then
\begin{align}
&\big\langle \mathcal{\widetilde D}_A\big(\rho_t^U,\omega^U\big) \big\rangle_U = \frac{2}{d_B^\Delta(d_B^\Delta+1)} \trx{\Pi_s C\otimes C}{{B}^{\otimes2}}  \nonumber  \\
&= \frac{1}{d_B^\Delta(d_B^\Delta+1)} \bigg(  \trx{(\id_B^\Delta)^{\otimes 2} C\otimes C}{{B}^{\otimes2}}  \nonumber \\
& \qquad \qquad \qquad \qquad \qquad +   \trx{\SWAP C\otimes C}{{B}^{\otimes 2}}  \bigg)   \nonumber   \\
&= \frac{1}{d_B^\Delta(d_B^\Delta+1)} \Big(  \trb{C}^2  +   \trb{C^2}  \bigg). \label{eq3:appaux1}
\end{align}

The operator $C$ is of the form $C = \trs{O}$, with $O$ Hermitian and acting only on the microcanonical window, from its definition above. Any such operator can be written as
\begin{equation}
O = \sum_{j =0}^{d_S^2-1} \sum_{k=0}^{\left(d_B^\Delta\right)^2-1} a_{jk} X_j Y_k
\end{equation}
where $a_{jk}$ are real coefficients, and $\{ X_j \}$ and $\{ Y_k \}$ are orthonormal bases of Hermitian operators on system and microcanonical window respectively~\cite{Short11}. They satisfy
\begin{align}
&\trs{X_j  X_{j'}} = \delta_{j j'} , \quad \forall  \{j,j'\} = \left(0,\dots d_S^2-1\right)  \\
&\trb{Y_k  Y_{k'}} = \delta_{k k'} , \quad \forall  \{k,k'\} = \left(0,\dots (d_B^\Delta)^2-1\right), \nonumber
\end{align}
$X_0 = \frac{\id_S}{\sqrt{d_S}}$ and $Y_0 = \frac{\id_B^\Delta}{\sqrt{d_B^\Delta}}$, while all other operators have trace $0$.
With these definitions we can write
\begin{align}
 &\trb{C^2} = \trb{\trs{O}\trs{O}} \nonumber \\
 &\qquad= \sum_{j j' = 0}^{d_S^2-1} \sum_{k k' = 0}^{\left( d_B^\Delta \right)^2-1} a_{jk} a_{j'k'} \trb{\trs{ X_j }\trs{ X_{j'} } Y_k Y_{k'}} \nonumber \\
 &\qquad= \sum_{k k' = 0}^{\left(d_B^\Delta\right)^2-1} a_{0k} a_{0k'} \trb{ Y_k Y_{k'}} \trs{ X_0 } \trs{ X_0 } \nonumber \\
 &\qquad= d_S \sum_{k k' = 0}^{\left(d_B^\Delta\right)^2-1} a_{0k} a_{0k'} \delta_{k k'} = d_S \sum_{k = 0}^{\left(d_B^\Delta\right)^2-1} a_{0k}^2 \nonumber \\
 &\qquad\le d_S \sum_{j = 0}^{d_S^2-1} \sum_{k = 0}^{\left(d_B^\Delta\right)^2-1} a_{jk}^2  = d_S \tr{O^2}.
\end{align}
From this result, and the definition of $C$ above,
\begin{align}
 &\trb{C^2} \le  \frac{d_S}{4\|A\|^2} \tr{(\rho_S \otimes \id_B^\Delta \big( A(t)-A_{eq} \big))^2}  \nonumber \\
 &\quad\le  \frac{d_S}{4\|A\|^2} \sqrt{ \tr{(\rho_S \otimes \id_B^\Delta) \big( A(t)-A_{eq} \big)^2 (\rho_S \otimes \id_B^\Delta)}  } \nonumber \\
 &\qquad \quad \times \sqrt{ \tr{\big( A(t)-A_{eq} \big) (\rho_S \otimes \id_B^\Delta)^2 \big( A(t)-A_{eq} \big)}} \nonumber \\
 &\qquad \quad\quad \le  \frac{d_S}{4\|A\|^2} \tr{(\rho_S \otimes \id_B^\Delta)^2} \|A(t)-A_{eq}\|^2 \nonumber \\
 &\qquad \quad\quad\le d_S d_B^\Delta 
\end{align}
by using Cauchy-Schwarz inequality on line 2, the fact that for positive semidefinite operators $\tr{PQ} \le \tr{P} \|Q\|$ on line 3, and the triangle inequality plus $\|A_{eq}\|\le \|A(t)\| = \|A\|$ on the last line.

With this result we see that eq.~(\ref{eq3:appaux1}) turns into
\begin{align}
&\big\langle \mathcal{\widetilde D}_A\big(\rho_t^U,\omega^U\big) \big\rangle_U = \frac{1}{d_B^\Delta(d_B^\Delta+1)} \left(  \trb{C}^2  +   \trb{C^2}  \right) \nonumber \\
&\le \frac{1}{d_B^\Delta(d_B^\Delta+1)} \bigg(  \frac{1}{4\|A\|^2} \tr{\rho_S \otimes \id_B^\Delta \big( A(t)-A_{eq} \big)}^2  \nonumber \\
& \qquad \qquad \qquad \qquad \qquad \qquad \qquad +  d_S d_B^\Delta  \Bigg)    \nonumber \\
&= \frac{1}{d_B^\Delta(d_B^\Delta+1)} \Bigg( \frac{(d_B^\Delta)^2}{4\|A\|^2}  \tr{\rho_S\otimes\frac{\id_B^\Delta}{d_B^\Delta} \big( A(t)-A_{eq} \big)}^2  \nonumber \\
& \qquad \qquad \qquad \qquad \qquad \qquad \qquad +  d_S d_B^\Delta  \Bigg)    \nonumber \\
&= \frac{d_B^\Delta}{d_B^\Delta+1}  \frac{\tr{\rho_0 \big( A(t)-A_{eq} \big)}^2}{4\|A\|^2}  +   \frac{d_S}{d_B^\Delta+1},
\end{align}
where $\rho_0 \equiv \rho_S\otimes\frac{\id_B^\Delta}{d_B^\Delta}$. We can now shift the time dependence back to the state to get
\begin{align}
&\big\langle \mathcal{\widetilde D}_A\big(\rho_t^U,\omega^U\big) \big\rangle_U \le \frac{d_B^\Delta}{d_B^\Delta+1}  \frac{\tr{\left( \rho_t  - \omega  \right) A}^2}{4 \|A\|^2}  \nonumber \\
& \qquad \qquad \qquad \qquad \qquad+   \frac{d_S}{d_B^\Delta+1}      \nonumber \\
& \qquad \le \frac{\tr{\left( \rho_t  - \omega  \right) A}^2}{4\|A\|^2}  +   \frac{d_S}{d_B^\Delta} \nonumber \\
& \qquad = \mathcal{\widetilde D}_A\big(\rho_t ,\omega  \big)  +   \frac{d_S}{d_B^\Delta},
\end{align}
which proves our claim. The bound $\frac{1}{d_B^\Delta+1} \le \frac{1}{d_B^\Delta}$ is only for presentation reasons and does not change the result much since $d_B^\Delta \gg d_S >1 $ in the regime we are interested in.

\section{The distribution $p_\alpha$ for a spin ring}
\label{app:simulation}

In order to get a deeper grasp of the behaviour of $p_\alpha$ we simulated $L$ interacting spin $1/2$'s, with a Hamiltonian given by
\begin{align}
H = \Omega \sum_{\lambda = 1}^L \sigma_\lambda^z + \gamma \Omega \sum_{\lambda = 1}^{L} \sigma_\lambda^x \otimes \sigma_{\lambda + 1}^x,
\end{align}
where $\sigma_\lambda^z$ and $\sigma_\lambda^x$ are the Pauli $z$ and $x$ operators for the spin $\lambda$, and we adopt the notation $\sigma_{L + 1}^x = \sigma_1^x$. 
The spin $\lambda = 1$ is taken to represent the system $\mathcal{S}$, and we focus on an observable $A_x$ and initial state $\rho$ given by  
\begin{align}
A_x = \sigma_1^x \otimes \bigotimes_{\lambda = 2}^L \id_\lambda, \qquad \rho_0 = \ket{1} \bra{1} \otimes \bigotimes_{\lambda = 2}^L \frac{\id_\lambda}{2},
\end{align}
the latter representing a bath in a maximally mixed initial state and the system in the eigenvector $\ket{1}$ of $\sigma_1^z$ with eigenvalue $1$. 

Figure~\ref{fig:physicaltimescales-simulation} depicts the normalized
distribution $p_\alpha = \tfrac{1}{Q} \tfrac{ \big| \rho_{jk} A_{kj} \big|  }{\|A_x\|}$ as a function of the energy gaps $G_\alpha = (E_j - E_k)$ as the 
number of spins increases, illustrating the transition between a distribution with several distinct peaks and a unimodal distribution.
In Figure~\ref{fig:physicaltimescales-simulation-adelta} we plot $a(\epsilon)$ and $\delta(\epsilon)$ for different values of $L$, illustrating their decrease with increasing $L$, for most values of $\epsilon$. 
Thus, even for moderate sizes of the bath, one can find an interval $\epsilon$ such that $\delta(\epsilon) \ll 1$ and $a \sim 1$. 
The first condition is necessary for Theorem~\ref{thm:simplethem} and the subsequent results to imply that equilibration occurs, while the second condition is necessary to  ensure that the equilibration timescale
does not grow for increasing bath sizes (see discussion after Theorem~\ref{theorem-boundmc}). 

For example, in this model we find that for $L = 9$ one can take $\epsilon$ such that $a(\epsilon) = 1$ and $\delta(\epsilon) \approx 0.006$. Then Theorem~\ref{thm:simplethem} gives
\begin{align}
      \left\langle \mathcal{\widetilde D}_{A} \left( \rho_t ,\omega  \right) \right\rangle_T
&\le \frac{ 2^{5/4} \pi  \|A_S\|^{1/2}  }{T \sqrt{ \bigg| \textnormal{Tr}\Big( \Big[ \big[\rho_0,H_S\!+\!H_I \big] , H_S\!+\! H_I \Big] A_S \Big) \bigg|} } \nonumber \\
      &\quad + 0.04, 
  \end{align}
with an upper bound on the equilibration time scale 
\begin{equation} 
 T_{\text{eq}} \equiv \frac{2^{5/4} \pi {\|A_S\|}^{1/2} }{\sqrt{\bigg| \text{Tr}\Big( \Big[ \big[\rho_0,H_S\!+\! H_I \big] , H_S\!+\! H_I\Big] A_S \Big) \bigg|}}, 
\end{equation}
dictated by the observable, initial state and system and interaction Hamiltonians, which is straightforward to calculate.

\begin{figure*}
\centering
\textbf{Distribution $p_\alpha$ for observable $A_x$ on translationally invariant spin ring 
\
\vspace{8pt}
\\
}
\begin{subfigure}{.5\textwidth}
  \centering
  \includegraphics[width=1\textwidth]{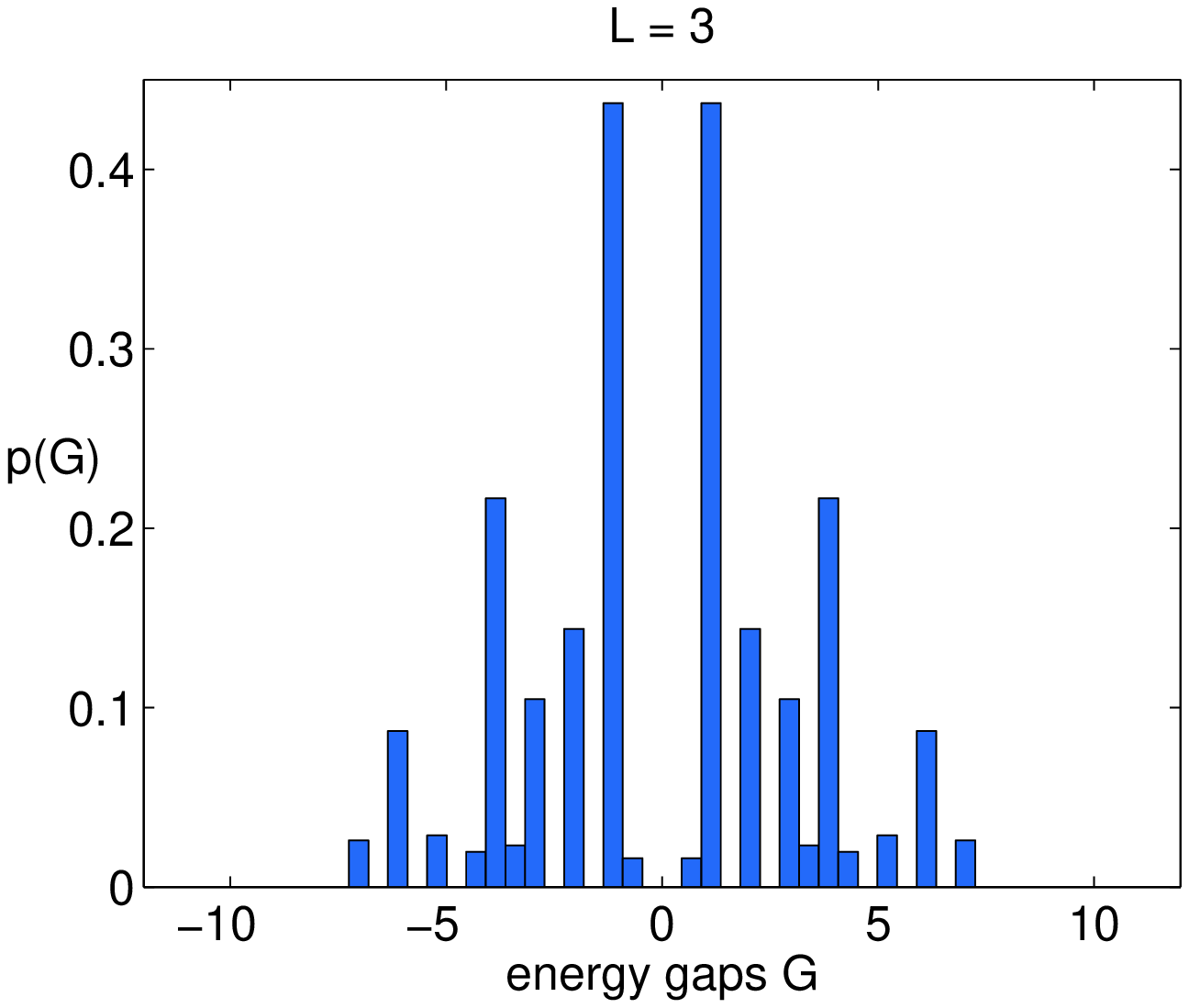}
\end{subfigure}%
\begin{subfigure}{.5\textwidth}
  \centering
  \includegraphics[width=1\textwidth]{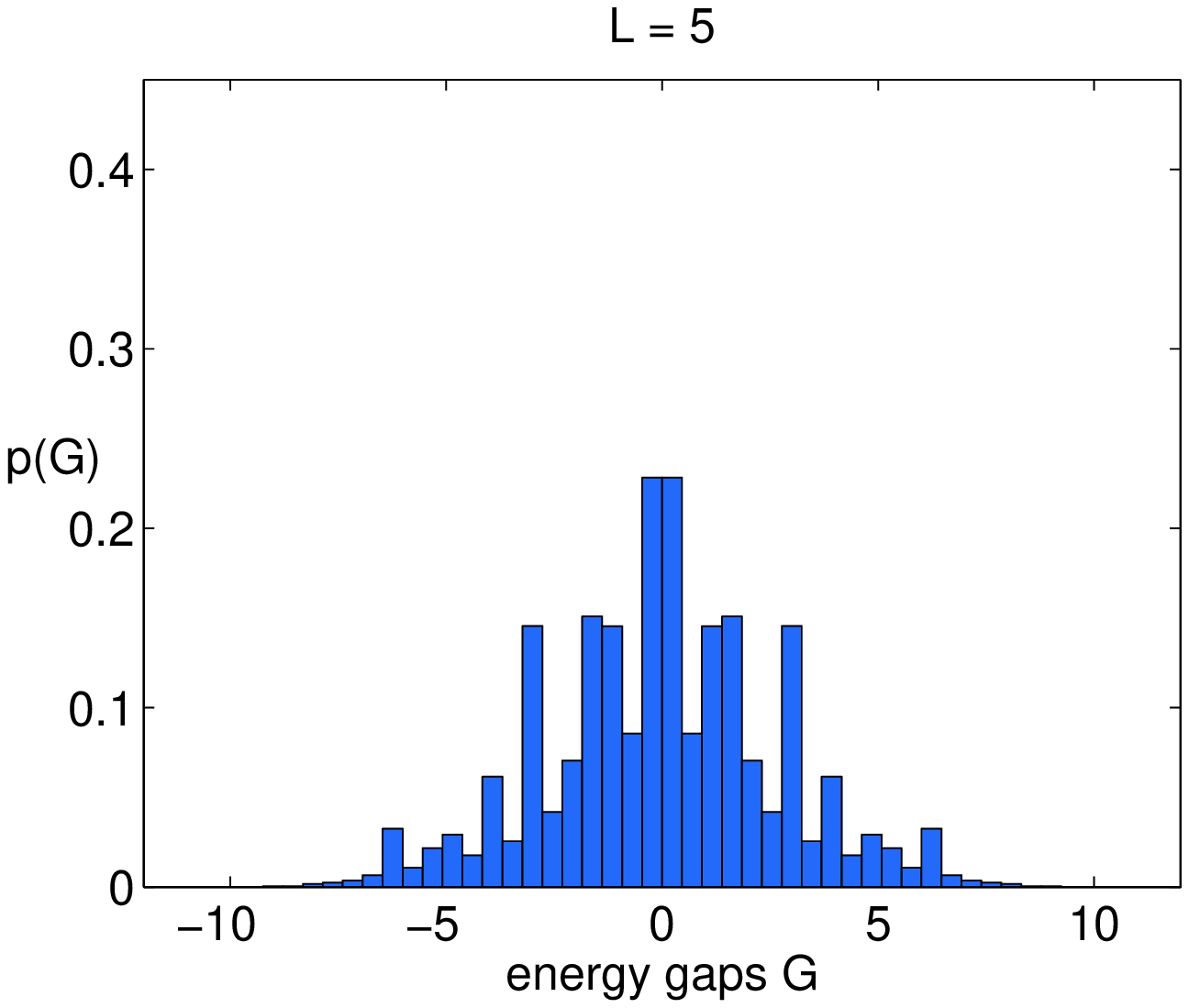}
\end{subfigure}
\\
\vspace{11pt}
\centering
\begin{subfigure}{.5\textwidth}
  \centering
  \includegraphics[width=1\textwidth]{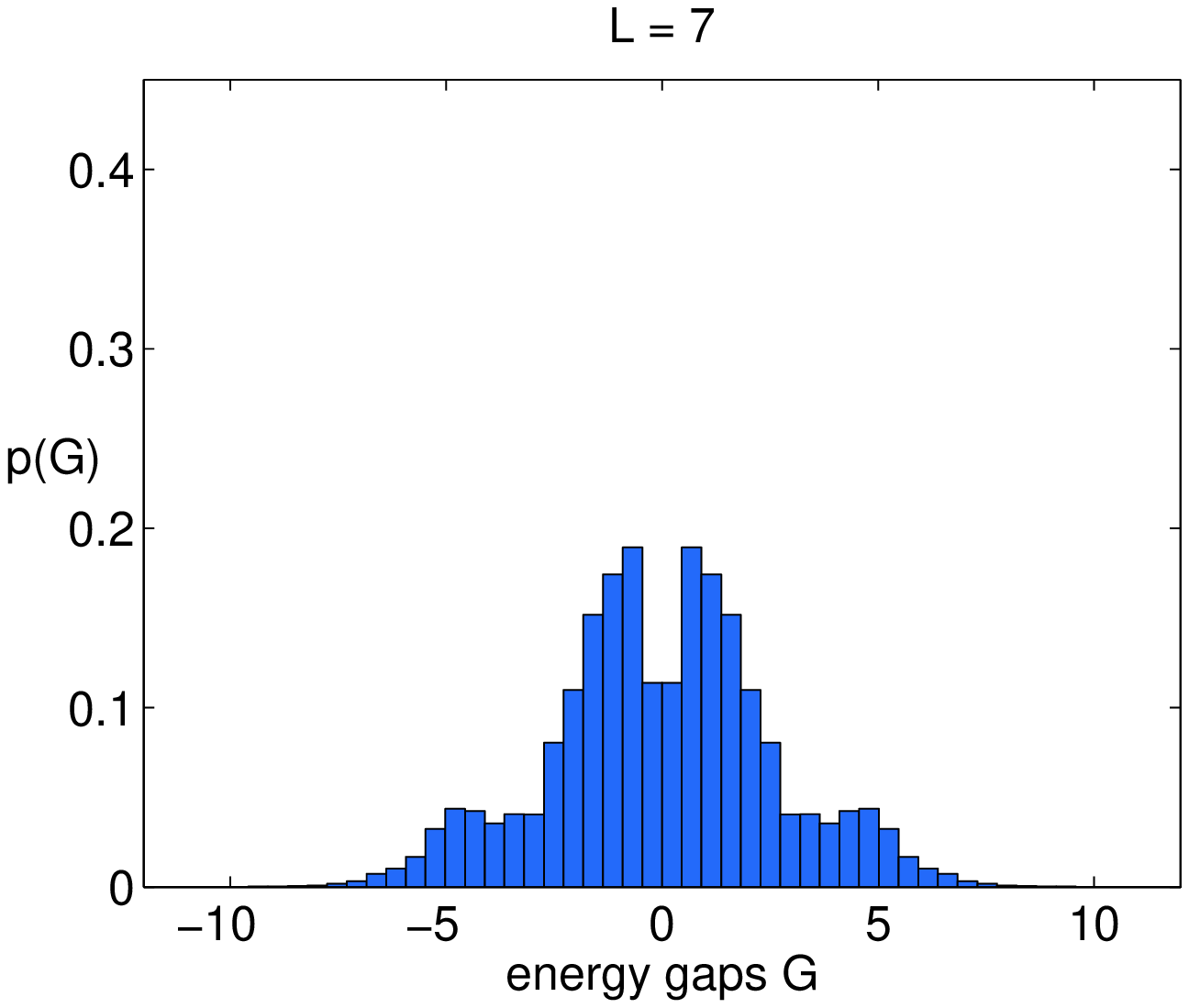}
\end{subfigure}\begin{subfigure}{.5\textwidth}
  \centering
  \includegraphics[width=1\textwidth]{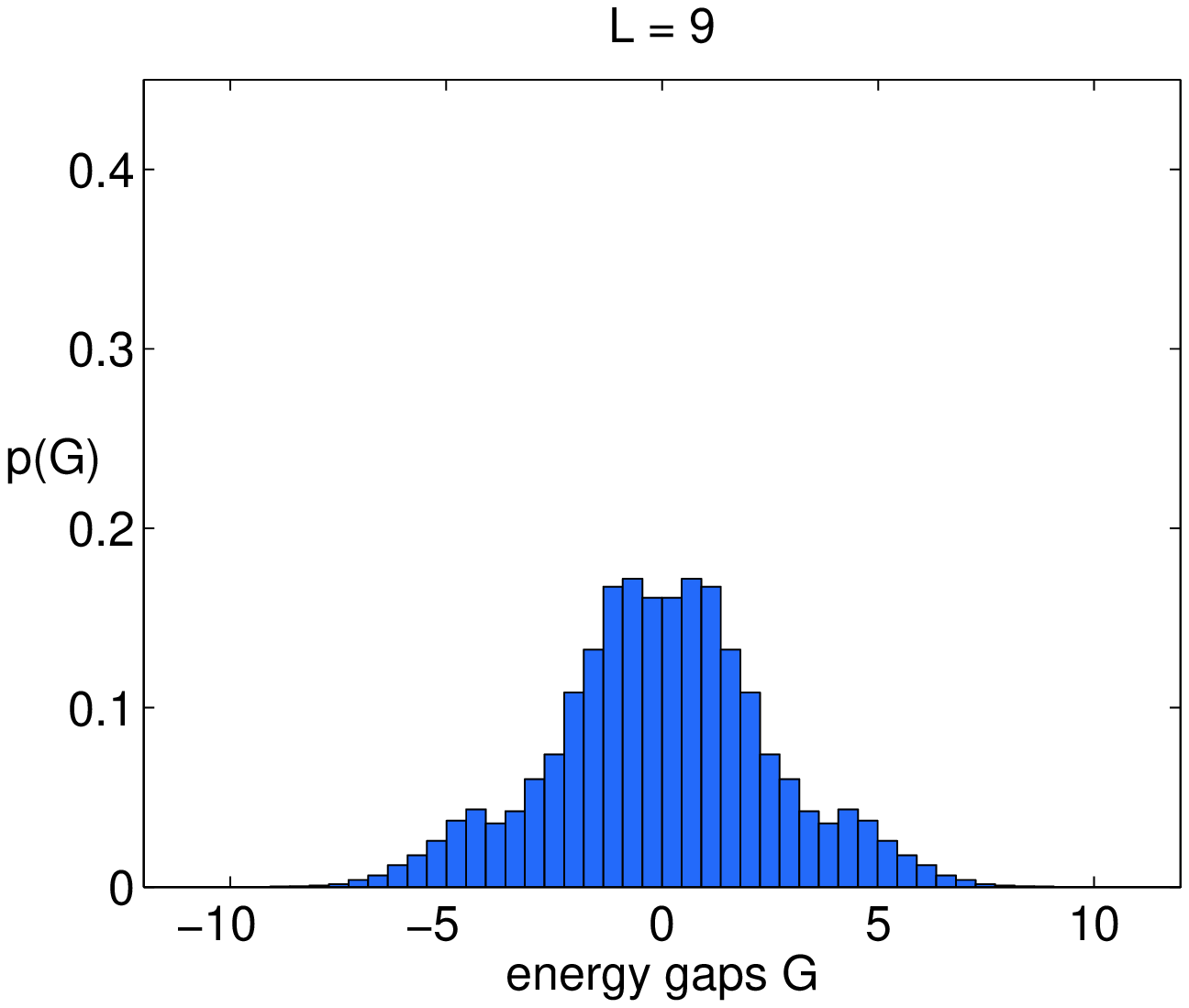}
\end{subfigure}
\caption{ 
Normalized histogram of the distribution $p_\alpha$ for an observable $A_x = \sigma_1^x \otimes  \bigotimes_{\lambda = 2}^L \id_\lambda$ and initial state $\rho_0 = \left| 1\right\rangle \left\langle 1 \right| \otimes \bigotimes_{\lambda = 2}^L \id_\lambda/2$ as a function of the energy gaps for a spin ring with coupling strength $\gamma = 1.1 \Omega$,
for increasing number $L$ of spins. 
For small $L$ the distribution $p_\alpha$ is composed of distinct peaks.
On the other hand, $p_\alpha$ is spread over more values as the size of the system increases. At the same time, as $L$ increases the distribution becomes more distinctly unimodal.
}
\label{fig:physicaltimescales-simulation}
\end{figure*}

\begin{figure*}
\centering
\textbf{Functions $a(\epsilon)$ and $\delta(\epsilon)$ for observable $A_x$ on translationally invariant spin ring
\
\vspace{8pt}
\\
}
\begin{subfigure}{.5\textwidth}
  \centering
  \includegraphics[width=1\textwidth]{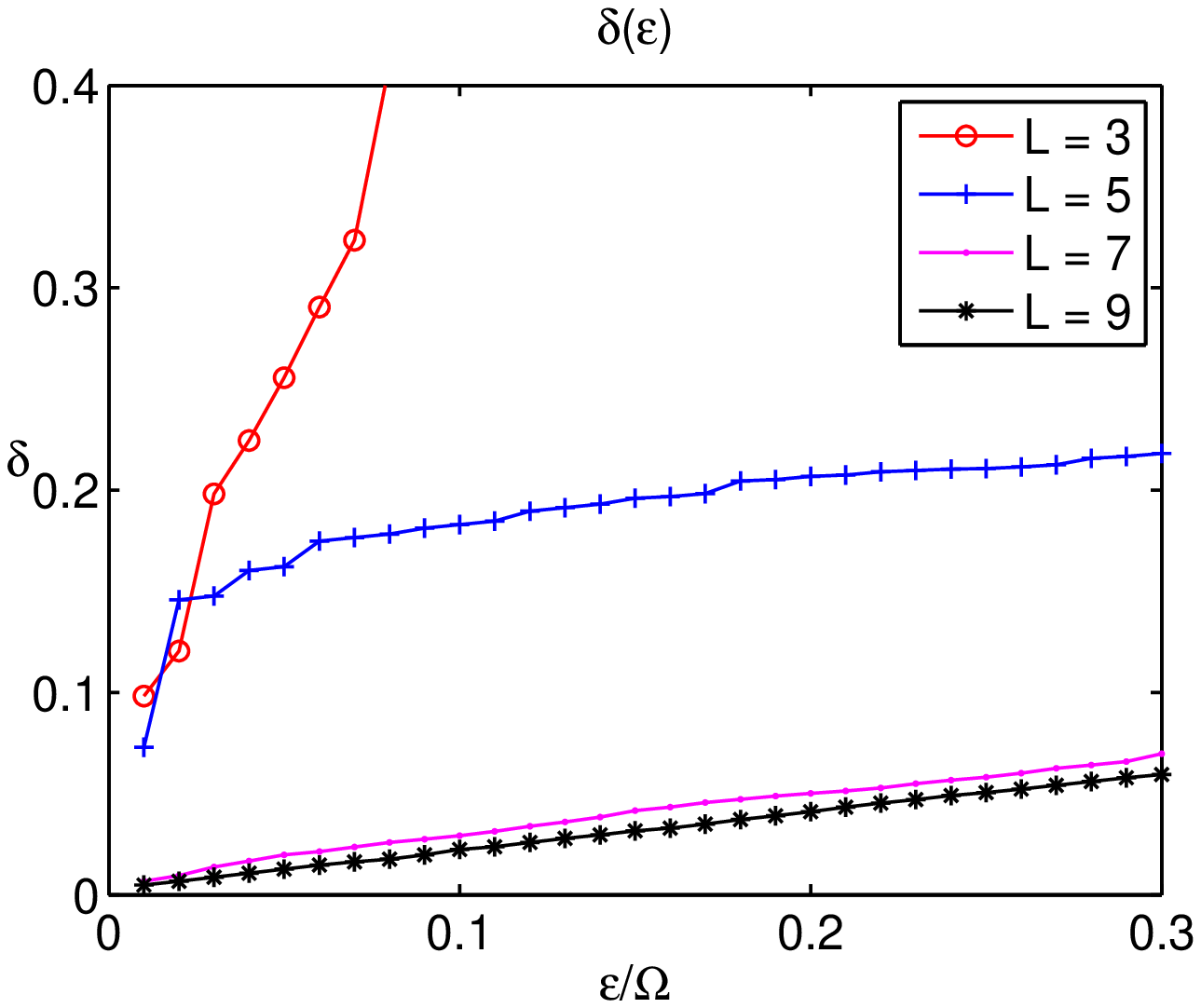}
\end{subfigure}%
\begin{subfigure}{.5\textwidth}
  \centering
  \includegraphics[width=1\textwidth]{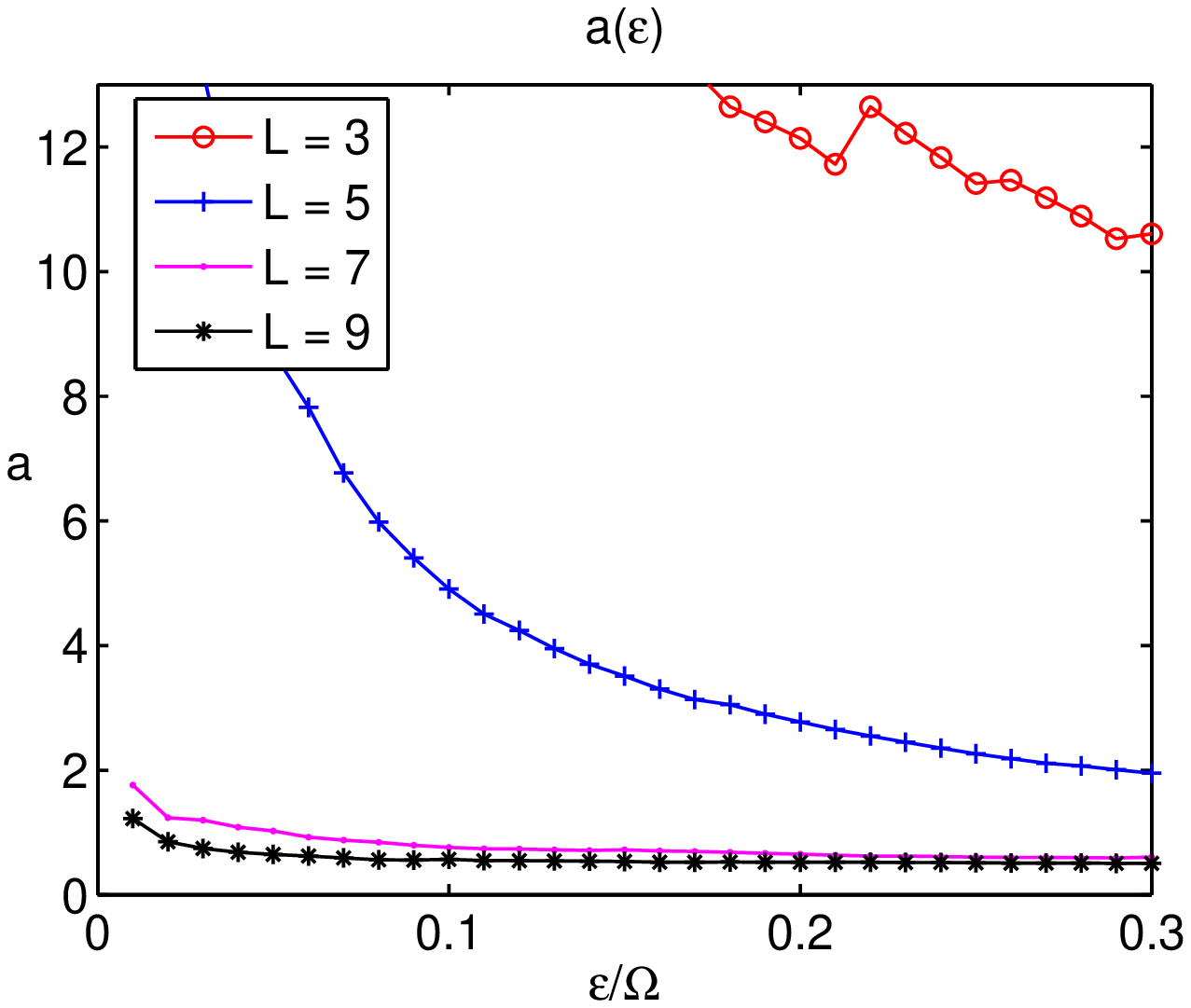}
\end{subfigure}
\caption{ 
Plots of $\delta(\epsilon)$ and $a(\epsilon)$, defined in Proposition~\ref{proposition-xidependence}, for an observable $A_x = \sigma_1^x \otimes \bigotimes_{\lambda = 2}^L \id_\lambda$ and initial state $\rho_0 = \left| 1\right\rangle \left\langle 1 \right| \otimes \bigotimes_{\lambda = 2}^L \id_\lambda/2$, for a spin ring with coupling constant $\gamma = 1.1 \Omega$. 
There is a trade-off between having small $\delta(\epsilon)$ and small $a(\epsilon)$.
However, for a fixed energy gap interval $\epsilon$ both $a(\epsilon)$ and $\delta(\epsilon)$ decrease as $L$ increases. 
Hence, as $L$ increases it becomes possible to find $\epsilon$ such that $a \sim 1$ and $\delta \ll 1$. 
Indeed, for $L = \{3, 5, 7, 9\}$ we have $a(\epsilon) \approx 1$ and $\delta(\epsilon) \approx \{1, 0.62, 0.02, 0.006\}$ for $\epsilon \approx \{3.26, 1.65, 0.05, 0.02\} \Omega$, respectively.
}
\label{fig:physicaltimescales-simulation-adelta}
\end{figure*}

Figure~\ref{fig:physicaltimescales-simulation-adeltaRandom}
 shows similar behaviour of $a(\epsilon)$ and $\delta(\epsilon)$ for a spin ring with random couplings.
We simulated a Hamiltonian 
\begin{align}
H = \Omega \sum_{\lambda = 1}^L \sigma_\lambda^z +  \Omega \sum_{  \lambda = 1}^{L} K_{\lambda} \sigma_\lambda^x \otimes \sigma_{\lambda + 1}^x,
\end{align}
where the couplings $K_\lambda$ are drawn at random from a Gaussian distribution with mean $\gamma$ and standard deviation $w$. 
For completeness, in figure~\ref{fig:physicaltimescales-simulation-adeltaRandom} we focus this time on an observable $A_z$, and the same initial state as above:
\begin{align}
A_z = \sigma_1^z \otimes \bigotimes_{\lambda = 2}^L \id_\lambda, \qquad \rho_0 = \ket{1} \bra{1} \otimes \bigotimes_{\lambda = 2}^L \frac{\id_\lambda}{2}.
\end{align}

\begin{figure*}
\centering
\textbf{Functions $a(\epsilon)$ and $\delta(\epsilon)$ for observable $A_z$ on a spin ring with random couplings
\
\vspace{8pt}
\\
}
\begin{subfigure}{.5\textwidth}
  \centering
  \includegraphics[width=1\textwidth]{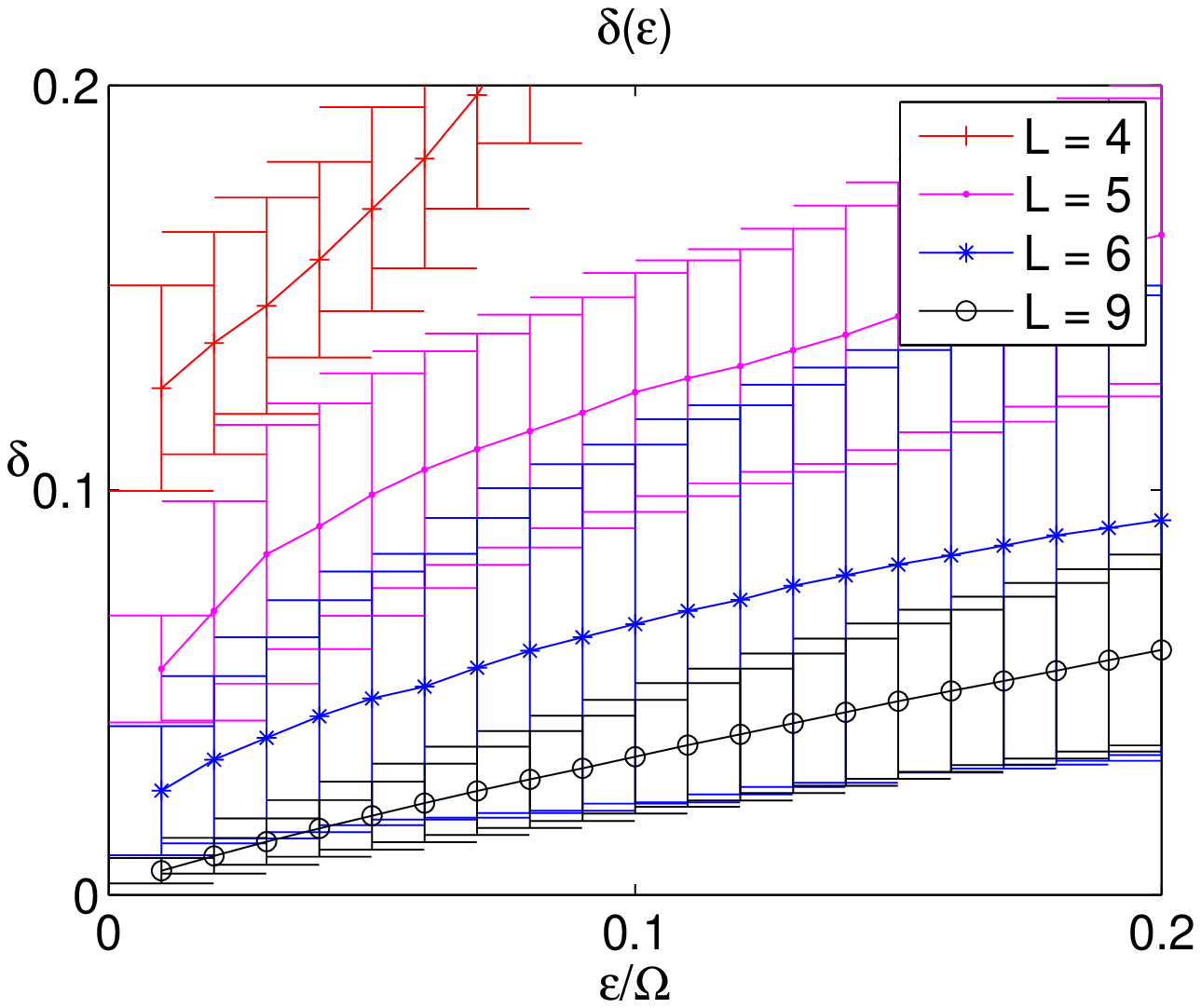}
\end{subfigure}\begin{subfigure}{.5\textwidth}
  \centering
  \includegraphics[width=1\textwidth]{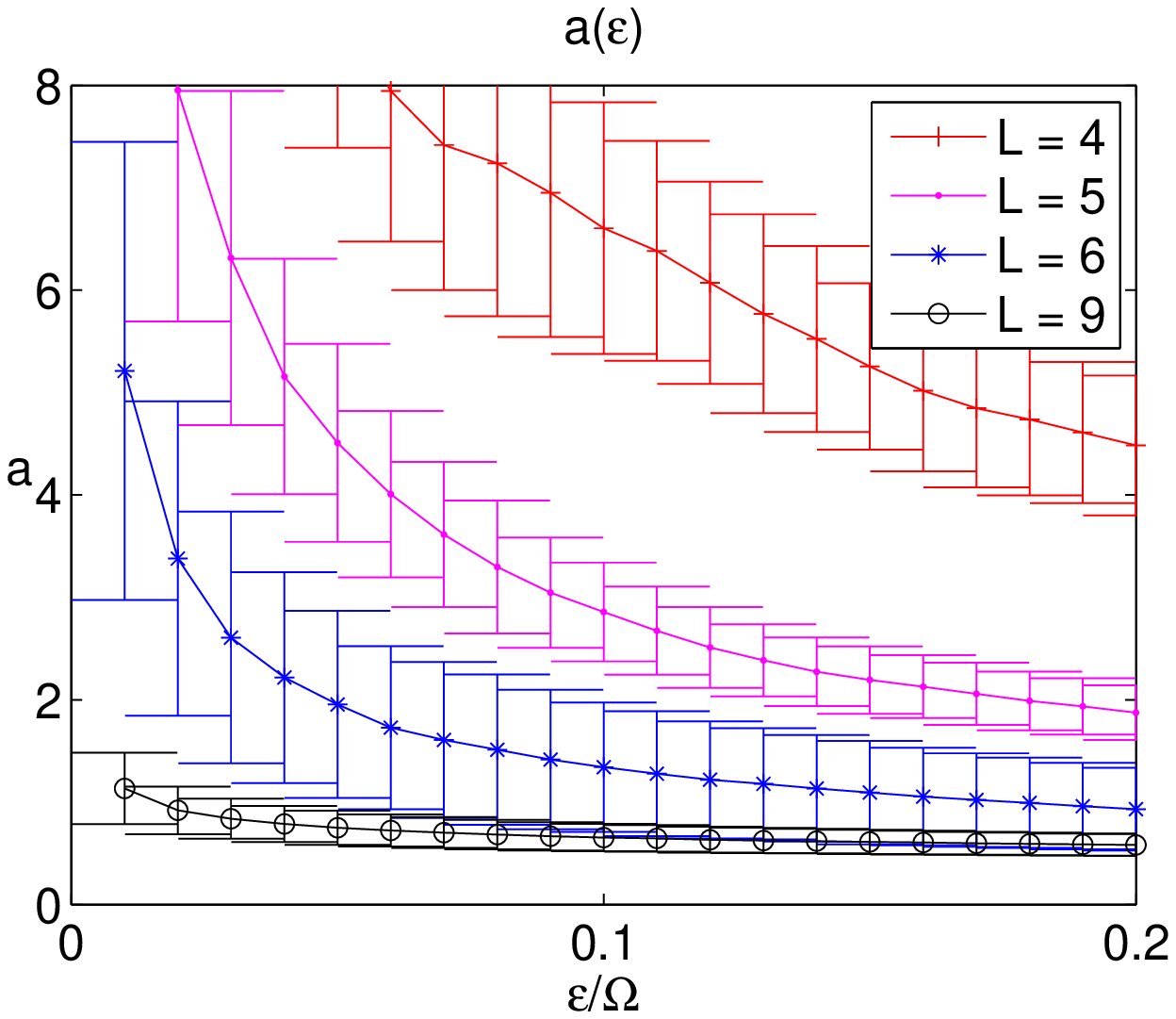}
\end{subfigure}
\caption{ 
Plots of average $\delta(\epsilon)$ and average $a(\epsilon)$ and their respective standard deviations (after 1000 realizations) for an observable $A_z = \sigma_1^z \otimes  \bigotimes_{\lambda = 2}^L \id_\lambda$ and initial state $\rho_0 = \left| 1\right\rangle \left\langle 1 \right| \otimes \bigotimes_{\lambda = 2}^L \id_\lambda/2$, for a spin ring with random couplings $K_\lambda$ with mean $\gamma = \Omega$ and standard deviation $w = 0.2 \Omega$. 
Once again, both $a(\epsilon)$ and $\delta(\epsilon)$ decrease as $L$ increases, for a fixed energy gap interval $\epsilon$. 
Hence, as $L$ increases it becomes possible to find $\epsilon$ such that $a \sim 1$ and $\delta \ll 1$. 
}
\label{fig:physicaltimescales-simulation-adeltaRandom}
\end{figure*}

\end{document}